\documentclass[runningheads]{llncs}
\usepackage[utf8]{inputenc}
\usepackage{hyperref}
\usepackage[ruled,noend,vlined,linesnumbered,]{algorithm2e}
    \DontPrintSemicolon{}
    \SetArgSty{textup}
    \SetProgSty{textup}
\usepackage[english]{babel}
\usepackage{braket}
\usepackage{tikz}
\usepackage{pgfplots}
\usepackage{mathtools}
\usepackage{booktabs}
\usepackage{amssymb}
\usepackage{breakcites}
\usepackage[capitalize,noabbrev]{cleveref}
\usepackage[inline]{enumitem}
\usepackage{scrextend} 

\definecolor{DarkPurple}{HTML}{332288}
\definecolor{DarkBlue}{HTML}{6699CC}
\definecolor{LightBlue}{HTML}{88CCEE}
\definecolor{DarkGreen}{HTML}{117733}
\definecolor{DarkRed}{HTML}{661100}
\definecolor{LightRed}{HTML}{CC6677}
\definecolor{LightPink}{HTML}{AA4466}
\definecolor{DarkPink}{HTML}{882255}
\definecolor{LightPurple}{HTML}{AA4499}

\definecolor{DarkBrown}{HTML}{604c38}
\definecolor{DarkTeal}{HTML}{23373b}
\definecolor{LightBrown}{HTML}{EB811B}
\definecolor{LightGreen}{HTML}{14B03D}

\definecolor{DarkOrange}{HTML}{FFDD00}

\pgfplotsset{width=1.0\textwidth,
  height=0.4\textwidth,
  cycle list={%
    solid,LightGreen,thick\\%
    dotted,LightRed,very thick\\%
    dashed,DarkBlue,thick\\%
    dashdotted,DarkPink,thick\\%
    dashdotdotted,LightGreen,thick\\%
    loosely dotted,very thick\\%
    loosely dashed,DarkBlue,very thick\\%
    loosely dashdotted,DarkPink,very thick\\%
    \\%
    DarkBrown,thick\\%
  },
  legend pos=north west,
  legend style={fill=none},
  legend cell align={left}}

\def\isanonymous{0}
\usepackage{ifthen}
\newcommand{\anonymous}[2]{
\ifthenelse{\equal{\isanonymous}{1}}%
{{#1}}%
{{#2}}%
}%

\usepackage{xcolor}
\definecolor{oxygenorange}{HTML}{FFDD00}
\usepackage[color=oxygenorange]{todonotes}

\usepackage{graphicx}

\usepackage[lambda,landau,operators,probability,sets,complexity,asymptotics]{cryptocode}

\usepackage{mathtools}

\DeclarePairedDelimiter{\bracket}{[}{]}
\DeclareMathOperator{\cL}{\mathcal{L}}

\DeclareMathOperator{\LWE}{LWE}

  \newcommand{\mat}[1]{\ensuremath{\mathbf{#1}}\xspace}
\renewcommand{\vec}[1]{\ensuremath{\mathbf{#1}}\xspace}

\newcommand{\Exp}[2]{\mathbb{E}_{#1}\bracket*{#2}} 	

\DeclareMathOperator{\polylog}{polylog}

\newcommand{\Prob}[2]{\Pr_{#1}\bracket*{#2}}

\newcommand{\round}[1]{\ensuremath{\left\lfloor{#1}\right\rceil}\xspace}

\newcommand{\DGauss}[2]{D_{#1\ifthenelse{\equal{#1}{}}{}{,#2}}}

\newcommand\numberthis{\addtocounter{equation}{1}\tag{\theequation}}

\renewcommand{\epsilon}{\varepsilon}

\newcommand{\rest}{\mathrm{lat}}
\newcommand{\fft}{\mathrm{fft}}
\newcommand{\enum}{\mathrm{enum}}

\usepackage{listings}
\lstdefinelanguage{Sage}[]{Python}{morekeywords={True,False,sage,cdef,cpdef,ctypedef,self},sensitive=true}

\ifxetex%
\usepackage{fontspec}
\else%
\usepackage{inconsolata}
\fi%

\usepackage[color={.9 .2 .2}]{attachfile2}

\hypersetup{
  colorlinks=false,
  urlbordercolor={.9 .2 .2},
  citebordercolor={.9 .2 .2},
  pdfborderstyle={/S/U/W 1},
}

\newif\ifeprint
\eprinttrue

\ifeprint 

\else 

\fi

\allowdisplaybreaks

\title{Quantum Augmented Dual Attack}
\ifeprint
\author{Martin R. Albrecht \and Yixin Shen}
\institute{Royal Holloway, University of London. \\ \texttt{\{martin.albrecht,yixin.shen\}@rhul.ac.uk}}
\else
\author{}
\institute{}
\fi

\begin{document}

\maketitle

\begin{abstract}
  We present a quantum augmented variant of the dual lattice attack on the Learning with Errors (LWE) problem, using classical memory with quantum random access (QRACM). Applying our results to lattice parameters from the literature, we find that our algorithm outperforms previous algorithms, assuming unit cost access to a QRACM\@.
  On a technical level, we show how to obtain a quantum speedup on the search for Fast Fourier Transform (FFT) coefficients above a given threshold by leveraging the relative sparseness of the FFT and using quantum amplitude estimation. We also discuss the applicability of the Quantum Fourier Transform in this context. Furthermore, we give a more rigorous analysis of the classical and quantum expected complexity of guessing part of the secret vector where coefficients follow a discrete Gaussian (mod \(q\)).
\end{abstract}


\section{Introduction}

The Learning With Errors (LWE) problem was introduced by Regev~\cite{Regev05} and has since become a major ingredient for constructing basic and more advanced cryptographic primitives. It asks to find \(\vec{s}\) given \((\mat{A},\vec{b})\) with \(\vec{b} \equiv \mat{A}\cdot \vec{s} + \vec{e} \bmod q\) where both \(\vec{s}\) and \(\vec{e}\) have small entries. Its conjectured hardness against quantum computers further makes all these constructions supposedly post-quantum. In NIST's Post Quantum Standardization Process, two out of four selected algorithms rely on the conjectured hardness of algebraic variants of Learning With Errors~\cite{AC:SSTX09,EC:LyuPeiReg10} problem.

From the perspective of a cryptanalyst equipped with a quantum computer, lattice problems such as LWE are frustrating. Known quantum speedups to solve these problems~\cite{PQCRYPTO:LaaMosvan13,PhD:Laarhoven15,AC:ChaLoy21} are tenuous at best~\cite{AC:AGPS20}. That is, while Grover's search offers a near quadratic quantum speedup for breaking, say, AES~\cite{EC:JNRV20} the gains against lattice problems are significantly more modest. This is due to the rich structure of the search space in lattice reduction algorithms that has given rise to refined structured search algorithms for these problems, e.g.~\cite{SODA:BDGL16}. As a consequence of this, the current state-of-the-art is that quantum algorithms on lattice problems can effectively be ignored when setting parameters.

The most efficient cryptanalysis techniques against LWE(-like) problems are ``primal'' and ``dual'' lattice attacks, named after whether lattice reduction is performed on the ``primal'' lattice related to \(\mat{A}\) or the ``dual'' lattice related to \(\{\vec{x} \in \ZZ_q^{m} \mid \vec{x} \cdot \mat{A} \equiv \vec{0} \bmod q\}\). Up until recently, dual attacks were generally considered less efficient for secrets \(\vec{s}\) drawn from a sufficiently wide distribution. Recent developments~\cite{AC:GuoJoh21,Matzov22} of dual attacks, however, have shown their ability to surpass primal attacks. These performance improvements are derived from combining lattice reduction on the scaled dual of a target lattice with an exhaustive search on a space related to the underlying secret \(\vec{s}\).  Roughly speaking, spending more resources on the exhaustive search part allows us to spend fewer resources on the lattice reduction part of the overall algorithm and vice versa.

In~\cite{AC:GuoJoh21} the search over part of the secret vector is realised using a Fast Fourier Transform style algorithm and the search space is significantly reduced by roughly considering only the most significant bits of this part of the secret. In~\cite{Matzov22} this last step is replaced by ``modulus switching''~\cite{FOCS:BraVai11,PKC:AFFP14,C:KirFou15,C:GuoJohSta15} which further provides significant performance gains. Overall, these newer iterations of the dual attack relate the search space to the underlying secret in such a way that large dimensions can now be covered even when the norm of the secret vector is not very small (previous versions of the dual attack relied on, say, coefficients \(\mathbf{s}_{i} \in \{-1,0,1\}\)).

Thus, with this new generation of dual attacks, unstructured search starts again to play a bigger role in costing attacks on LWE\@. It is therefore natural to ask what performance gains can be obtained by tackling this unstructured search using a Grover-like quantum algorithm.
More precisely,~\cite{Matzov22} relies on two different kinds of unstructured search:
\begin{itemize}
\item Secret guessing: part of the secret is exhaustively searched until a match is found. Since the secret is generated according to a discrete Gaussian of small width, a significant speedup can be obtained by starting the search with the most likely values of the secret first. The expected complexity of this step is known as the guessing complexity.
\item FFT threshold: given a list of values in a $n$-dimensional array, and a threshold, the problem is to decide whether one of the coefficients of the Fourier transform of the array is above the threshold. This problem arises when trying to determine whether the secret guess was correct by distinguishing between a uniform distribution and a Gaussian one.
\end{itemize}


\subsubsection{Contributions.}

After some preliminaries in \cref{sec:preliminaries}, we provide a quantum version of the dual attack of~\cite{Matzov22}. Specifically, our improvements are twofold.

In \cref{sec:guessing} we give a more rigorous analysis of the (classical and quantum) expected complexity of guessing a vector (whose coefficients are) drawn from a modular discrete Gaussian. In~\cite{Matzov22}, the authors estimated this complexity as the exponential of the entropy which is known not to be correct in general~\cite{Massey94}. We show that this complexity is indeed related to the entropy in the case of a (modular) discrete Gaussian, albeit up to an exponential factor in the dimension.

In \cref{sec:quantum-mean-estimation} we show how to obtain a quantum speedup on the search for Fast Fourier Transform (FFT) coefficients above a given threshold. This was left as an open problem in~\cite{Matzov22}. Here, we leverage the relative sparseness of the FFT and use amplitude estimation to estimate the Fourier coefficients.

In \cref{sec:quantum_dual_attack} we provide and analyse a quantum augmented dual attack utilising our guessing and mean estimation results.

In \cref{sec:application}, we then estimate the impact of our algorithm on the cost of solving instances of lattice-based schemes with parameters taken from the literature. We will refer to such parameters as ``lattice parameters'' going forward. Following the literature, we evaluate the complexity of our algorithm under the assumption of unit-cost access to a classical memory with quantum random access (QRACM).\footnote{\label{fn:bug}A previous version of this work reported wrong estimates due to a bug in our estimation code, which did not match the theorems given in the written report.} Interestingly, our optimisation routine produces attack parameters that skip the FFT step (``\(k_{\fft} = 0\)'') in favour of a more costly secret guessing stage (``\(k_{\enum} > 0\)''). Looking ahead, we speculate that this is because we get a super quadratic speedup on the enumeration part and only a quadratic speedup on \(p^{k_{\fft}}\). If \(2^{H(\chi_s)}\) is sufficiently smaller than \(p\) then it is better to enumerate the entire first \(n-k_{\rest}\) coordinates;  a case we are in often for practical parameters.

In \cref{sec:open_problem}, we discuss the FFT threshold problem and its quantum complexity.
Any significant speedup on this problem would yield major improvements in the complexity
of the dual attack. We argue that the Quantum Fourier Transform (QFT) does not seem applicable
in this context, despite being the natural approach.

\subsubsection{Discussion.}\label{sec:interpreting} On the one hand, we analyse our algorithm in the same ``cost model'' as prior work such as~\cite{PQCRYPTO:LaaMosvan13,PhD:Laarhoven15,AC:ChaLoy21} and in this cost model we obtain mild but noticeable speed-ups over previous work. These are mostly derived from our new algorithm for mean estimation and its composition with other known quantum algorithms from the literature such' as Grovers search for a secret and quantum versions of lattice-reduction attacks.

On the other hand, the cost model adopted here and prior work assume unit cost for quantum operations including accessing classical RAM in superposition (QRACM). As discussed in e.g.~\cite{AC:AGPS20}, this is a very strong assumption.

Furthermore, even in this model our algorithm falls far short of the quadratic speed-up we would need to ``force'' a revision of lattice parameters. This is because the resistance of post-quantum algorithms to quantum computers is routinely, e.g.~in the NIST PQC Standardization Process, compared to that of the AES family of block ciphers (or other symmetric-key primitives). Here, the state of the art is that AES-\(\lambda\) resists classical attacks of cost \(\approx 2^{\lambda}\) and quantum attacks of cost \(\approx 2^{\lambda/2}\), the latter being due to Grover's algorithm, see~\cite{EC:JNRV20} for more detailed cost estimates. Thus, parameters for post-quantum schemes are chosen such that they resist known classical attacks of cost \(\approx 2^{\lambda}\) and known quantum attacks of cost \(\approx 2^{\lambda/2}\) and any quantum algorithm with complexity \(\gg 2^{\lambda/2}\) will not affect the claimed security level. Our algorithm has cost \(\gg 2^{\lambda/2}\), whatever the cost model.

\section{Preliminaries}\label{sec:preliminaries}

Recall that \(e^{ix} = \cos (x) + i\, \sin(x)\). For any $z\in \CC$, we write $\Re(z)$ for its real part.  We write \([x,y]\) for the interval \(\{x, x+1, \ldots, y\} \subset \ZZ\). We denote matrices by bold uppercase letters, e.g.~\(\mat{A}\), and vectors by bold lowercase letters, e.g.~\(\vec{v}\). We treat vectors as column matrices. We write \(\vec{v}^{T}\) for the transpose of \(\vec{v}\).

For any $x\in\ZZ_q$, denote by $\tilde{x}\in x+q\ZZ$ the unique integer such that $|\tilde{x}|\leqslant\tfrac{q-1}{2}$. We extend this notion to vectors in $\vec{x}\in\ZZ_q^n$ componentwise. In other words, $\tilde{\vec{x}}$ is the lift from $\ZZ_q$ to $\ZZ$ centered on $0$. We define \(\norm{\vec{x}}\) for \(\vec{x} \in \ZZ_{q}^{n}\) as \(\norm{\tilde{\vec{x}}}\).

\subsection{Lattices}\label{sec:lattices}

A lattice \(\cL\) is a discrete subgroup of \(\RR^{d}\). We can represent it as \(\{\sum x_{i} \cdot \vec{b}_{i} | x_{i} \in \ZZ\}\) where \(\vec{b}_{i}\) are the columns of a matrix \(\mat{B}\), we may write \(\cL(\mat{B})\). If \(\mat{B}\) has full column rank, we call \(\mat{B}\) a basis.

While the central object of this work, the dual attack, critically relies on lattice reduction, such as the BKZ algorithm, we mostly make blackbox use of these algorithms here. Thus, we refer the reader to e.g.~\cite{AC:GuoJoh21,Matzov22} for details. In particular, the blackbox use we make of lattice reduction algorithms and, critically, lattice sieving algorithms is captured in \cref{alg:short_vectors_sampling}.

\begin{algorithm}[h]
    \KwIn{A basis $\mat{B} =\begin{bmatrix}\vec{b}_0&\ldots&\vec{b}_{d-1}\end{bmatrix}$ for a lattice and $2 \leq \beta_0,\beta_1 \in \ZZ \leqslant d$ and $D$.}
    \KwOut{A list of $D$ vectors from the lattice.}
     \(L\)=\{\}.\;
    \For{\(i \in [0, \lceil D/N_{\mathrm{sieve}}(\beta_1)\rceil -1]\)}{
    Randomise the basis \(\mat{B}\).\;
    Run BKZ-\(\beta_{0}\) to obtain a reduced basis \(\vec{b}'_{0}, \ldots, \vec{b}'_{d-1}\).\;
    Run a sieve in dimension \(\beta_{1}\) on the sublattice spanned by \(\vec{b}'_{0}, \ldots, \vec{b}'_{\beta_1-1}\) to obtain a list of \(N_{\mathrm{sieve}}(\beta_1)\) vectors and add them to \(L\).\;
  }
  \Return{\(L\)}
  \caption{Short Vectors Sampling Procedure~\cite{AC:GuoJoh21}}\label{alg:short_vectors_sampling}
\end{algorithm}

In \cref{alg:short_vectors_sampling} the BKZ-\(\beta_{0}\) call performs lattice reduction with parameter \(\beta_{0}\) where the cost of the algorithm scales at least exponentially with \(\beta_{0}\). The BKZ algorithm proceeds by making polynomially many calls to an SVP oracle. In this work, this oracle is instantiated using a lattice sieving algorithm which is also called explicitly in \cref{alg:short_vectors_sampling} with parameter \(\beta_{1}\). Such a sieving algorithm outputs \(N_{\mathrm{sieve}}(\beta_1)\) short vectors in the lattice \(\cL(\mat{B})\) and has a cost exponential in \(\beta_{1}\). The magnitude \(N_{\mathrm{sieve}}(\beta_1)\) also grows exponentially with \(\beta_{1}\) but slower than the cost of sieving.
We will write \(T_{\mathrm{BKZ}}(d,\beta_0)\) for the cost of running BKZ-\(\beta_{0}\) in dimension \(d\) and \(T_{\mathrm{sieve}}(\beta_1)\) for the cost of sieving in dimension \(\beta_{1}\). We may instantiate the lattice sieve with a classical algorithm~\cite{SODA:BDGL16} which has a cost of \(2^{0.292\,\beta_1 + o(\beta_1)}\). We may also instantiate the lattice sieve with a quantum augmented variant of sieving~\cite{PQCRYPTO:LaaMosvan13,PhD:Laarhoven15,AC:AGPS20,AC:ChaLoy21} which has a cost of \(2^{0.257\, \beta_1 + o(\beta_1)}\). Thus, according to the best known algorithms we have  \(T_{\mathrm{BKZ}}(d,\beta_0) \in \poly[d] \cdot 2^{\Theta(\beta_0)}\) and \(T_{\mathrm{sieve}}(\beta_1) \in 2^{\Theta(\beta_1)}\).


\subsection{Learning with Errors}
The Learning with Errors problem (LWE) is defined as follows.

\begin{definition}[LWE]\label{def:lwe}
  Let $n,m,q \in \NN$, and let $\chi_s,\chi_e$ be distributions over $\ZZ_q$. Denote by $\LWE_{n,m,\chi_s,\chi_e}$ the probability distribution on $\ZZ_q^{m\times n}\times \ZZ_q^m$ obtained by sampling the coordinates of the matrix $\mat{A} \in \ZZ_q^{m\times n}$ independently and uniformly over $\ZZ_q$, sampling the coordinates of $\vec{s} \in \ZZ_q^n$, $\vec{e} \in \ZZ_q^m$ independently from $\chi_s$ and $\chi_e$ respectively, setting \(\vec{b} \coloneqq \mat{A}\cdot \vec{s}+\vec{e} \bmod q\) and outputting $(\mat{A},\vec{b})$.
\end{definition}

We define two problems:
\begin{itemize}
    \item Decision-$\LWE$. Distinguish the uniform distribution over $\ZZ_q^{m\times n}\times \ZZ_q^m$ from
$\LWE_{n,m,\chi_s,\chi_e}$.
    \item Search-$\LWE$. Given a sample from $\LWE_{n,m,\chi_s,\chi_e}$, recover $\vec{s}$.
\end{itemize}


\subsubsection{Dual Attack.}\label{sec:dual-attack}

Dual-lattice attacks, or simply ``dual attacks'', on LWE and related problems were introduced in~\cite{PQCBook:MicReg09}. In its simplest form it proceeds as follows. Given either \((\mat{A},\mat{A}\cdot \vec{s} + \vec{e})\) or \((\mat{A}, \vec{u})\) where \((\mat{A},\vec{u})\) are uniform and w.l.o.g \(\vec{s},\vec{e}\) are short~\cite{C:ACPS09}, the attack finds short \(\vec{x}_j\) s.t.~\(\vec{x}_j^{T} \cdot \mat{A} \equiv \vec{0}\). Then, we either obtain \(\vec{x}_j^{T} \cdot \mat{A} \cdot \vec{s} + \left\langle  \vec{x}_j, \vec{e} \right\rangle = \left\langle  \vec{x}_j, \vec{e} \right\rangle\) or \(\left\langle \vec{x}_j, \vec{u} \right\rangle\). The former follows a distribution with small entries, i.e.~the distribution of \(|e_{j}|\) for \(e_{j} \coloneqq \left\langle \vec{x}_j, \vec{e} \right\rangle\) is biased towards elements \(< q/2\), and the latter follows a uniform distribution mod \(q\).

In~\cite{USENIX:ADPS16}, the ``normal form'' of the dual attack was introduced which finds short \(\vec{x}_j\) such that \(\vec{x}_j^{T} \cdot \mat{A} \equiv \vec{y}_j \bmod q\) with \(\vec{y}_j\) short. We then obtain
\[
\vec{x}_j^{T} \cdot \mat{A} \cdot \vec{s} + \left\langle  \vec{x}_j, \vec{e} \right\rangle = \left\langle \vec{y}_j, \vec{s} \right\rangle  + \left\langle  \vec{x}_j, \vec{e} \right\rangle,\] which follows a distribution with small entries when \(\vec{y}_j,\vec{s},\vec{x}_j\) and \(\vec{e}\) are short.

In~\cite{EC:Albrecht17} a composition of the dual attack with a guessing stage (and some scaling) was introduced with a focus on vectors \(\vec{s}\) that are sparse and small compared to \(\vec{e}\). The idea is to split \(\mat{A} = [\mat{A}_{0} \ \mat{A}_{1}]\) such that \[\vec{b} \equiv \mat{A}_{0}\cdot \vec{s}_{0} + \mat{A}_{1} \cdot \vec{s}_{1} + \vec{e} \bmod q.\] Then the dual attack is run on \(\mat{A}_{0}\) s.t.
\[
\left\langle \vec{x}_j,\vec{b} \right\rangle \equiv \vec{x}_j^{T} \cdot \mat{A}_{0} \cdot \vec{s}_{0} + \vec{x}_j^{T} \cdot \mat{A}_{1} \cdot \vec{s}_{1} + \left\langle \vec{x}_j,\vec{e} \right\rangle = \left\langle \vec{y}_j, \vec{s}_{0} \right\rangle + \vec{x}_j^{T} \cdot \mat{A}_{1} \cdot \vec{s}_{1} + \left\langle \vec{x}_j,\vec{e} \right\rangle.
\]
Thus, guessing the correct \(\vec{s}_{1}\) and computing \(\left\langle \vec{x}_j, \mathbf{b}  \right\rangle - \vec{x}^{T} \cdot \mat{A}_{1} \cdot \vec{s}_{1}\) produces a value that follows a distribution with small entries. In~\cite{INDOCRYPT:EspJouKha20} this was generalised to more general secret distributions paired with additional improvements on the exhaustive search over \(\vec{s}_{1}\). In~\cite{AC:GuoJoh21} further improvements were presented. In particular, the search over \(\vec{s}_{1}\) is realised using a Fast Fourier Transform style algorithm and the search space is significantly reduced by roughly considering only the most significant bits of \(\vec{s}_{1}\). In~\cite{Matzov22} this last step is replaced by ``modulus switching''~\cite{FOCS:BraVai11,PKC:AFFP14,C:KirFou15,C:GuoJohSta15} which provides significant performance gains.\footnote{Another significant gain reported in~\cite{Matzov22} is due to an improvement to the lattice sieving algorithm from~\cite{SODA:BDGL16} but discussing this is out of scope of this work.}

\subsection{Discrete Gaussian Distribution}\label{sec:discrete_gaussian}

Let $\sigma>0$. For any $\vec{x}\in\RR^d$, we let $\rho_\sigma(\vec{x}) \coloneqq \exp(-\norm{\vec{x}}^2/2\sigma^2)$. Note that this is different from the other (also commonly used) definition, where $\tfrac{1}{2}$ is replaced by $\pi$ in the exponent. This change is inconsequential to our results. We extend the definition of \(\rho_{\sigma}(\cdot)\) to sets of vectors \(\mathcal{S}\) by letting \(\rho_{\sigma}(\mathcal{S}) \coloneqq \sum_{\vec{x} \in \mathcal{S}} \rho_{\sigma}(\vec{x})\). For any lattice $\cL\subset\RR^d$, we denote by $D_{\cL,\sigma}$ the discrete Gaussian distribution over $\cL$, defined by $D_{\cL,\sigma}(\vec{x})=\rho_\sigma(\vec{x})/\rho_\sigma(\cL)$ for all $\vec{x}\in\cL$. Observing that $D_{\ZZ^n,\sigma}(\vec{x})$ only depends on $\norm{\vec{x}}$, we abuse notation and for \({\ell} = \norm{\vec{x}}\) write $\rho_\sigma({\ell})=\exp(-\ell^2/2\sigma^2)$ and $D_{\ZZ^n,\sigma}({\ell})=\rho_\sigma({\ell})/\rho_\sigma(\ZZ^n)$.

We will also make use of the modular discrete Gaussian. For any $q\in\NN$, we denote
by $D_{\ZZ_q^d,\sigma}$ the modular discrete Gaussian distribution over $\ZZ_q^d$ defined by
\[
    D_{\ZZ_q^d,\sigma}(\vec{x})=\frac{\rho_\sigma(\vec{x}+q\ZZ^d)}{\rho_\sigma(\ZZ^d)}.
\]
Note that the distribution $D_{\ZZ_q^d,\sigma}$ is isomorphic to the distribution
$D_{\ZZ_q,\sigma}^d$, a fact that we will use often implicitly.

We let \(\Phi(x)=\frac{1}{\sqrt{2 \pi}} \int_{-\infty}^{x} \exp({-t^{2} / 2})~\mathrm{d}t\) be the cumulative distribution function (cdf) of the standard normal distribution and \(\Phi^{-1}(x) : [0,1) \rightarrow \mathbb{R}\) its inverse.

\subsection{Quantum Computing}

\subsubsection{Quantum Circuit Model.} In the quantum circuit model, the time complexity is the circuit size, which is the total number of elementary quantum gates. The space complexity is the number of qubits used. We will assume that the elementary quantum gates come from a fixed universal set. Up to constant factors, the complexity does not depend on the universal set that we have chosen.
Since all unitary transforms are invertible, any quantum circuit $\mathcal{A}$ is reversible
and we denote by $\mathcal{A}^\dagger$ its inverse, which is also equal to its conjugate transpose
when viewed as a matrix.

Given a function $f:\bin^n \rightarrow \bin^m$, we say that a quantum circuit, implementing a unitary $U$ that acts on $n+\ell+m$ qubits, \emph{computes $f$ with probability} $\alpha$ if for every $x$, a measurement on the last $m$ qubits of $U\ket{x}\ket{0^\ell}\ket{0^m}$ outputs $f(x)$ with probability at least $\alpha$. The exact location of the qubits that we measure for the output actually does not matter, since we can also apply SWAP gates (implementable by elementary gates) to swap them to the $m$ last positions. The extra $\ell$ qubits that are not part of the input/output are called \emph{ancilla qubits} (or work space).

\subsubsection{Quantum Query Model.} We use the standard form of the \emph{quantum query model}: given a unitary $\mathcal{O}$, we say that a circuit computes $f$ with \emph{oracle access} to $\mathcal{O}$ if by augmenting the model with the unitary $\mathcal{O}$, we can construct a circuit computing $f$. The number of \emph{queries} on $\mathcal{O}$ is the number of unitary $\mathcal{O}$ in the circuit. If we find an efficient algorithm for a problem in query complexity and we are given an explicit circuit realizing the black-box transformation of the oracle $\mathcal{O}$, we will have an efficient algorithm for an explicit computational problem.

\subsubsection{Quantum Algorithms.}
We say that a \emph{quantum algorithm computes a function} $F:\bin^*\to\bin^*$ \emph{with probability} $\alpha$ if there is a classical algorithm $\mathcal{A}$ \emph{with quantum evaluation} that outputs $F(w)$ with probability $\alpha$ on input $w$. By quantum evaluation we mean that the algorithm can, any number of times during the computation, build a quantum circuit and evaluate it, that is measure the state $U\ket{0}$ where $U$ is the unitary implemented by the circuit.

For a family of algorithms parameterised by \(n\), the \emph{time complexity} $T(n)$ is the classical time complexity of $\mathcal{A}$ plus the time complexity of the circuits, i.e.~the number of gates. The \emph{classical space complexity} $S(n)$ is the space complexity of $\mathcal{A}$ (ignoring quantum evaluations). The \emph{quantum space complexity} $Q(n)$ is the maximal space complexity of all circuits, i.e.~the maximum number of qubits used. In the natural way, we say that a quantum algorithm has oracle access to $\mathcal{O}$ if it produces circuits with oracle access to $\mathcal{O}$. The query complexity of the algorithm is the sum of the query complexity of the circuits.
Now that the quantum model of computation is properly defined, we can express the fact that every reversible classical computation can be implemented by a quantum computer, although at a non-negligible cost. Here reversible classical computation means that there is a one-to-one mapping between the inputs and the outputs and that a transform exists to return from the output to the input.

\begin{theorem}[\cite{Bennett89,LS90}]\label{th:irreversible_to_reversible}
    Given any $\epsilon>0$ and any classical computation with running time $T$ and space complexity $S$,
    there exists an equivalent reversible classical computation with running time $O(T^{1+\epsilon}/S^\epsilon)$
    and space complexity $O(S(1+\ln(T/S)))$.
\end{theorem}

\begin{corollary}\label{cor:irreversible_to_quantum}
  Given any $\epsilon>0$ and any classical computation with running time $T$ and space complexity $S$, there exists an equivalent quantum circuit of size $O(T^{1+\epsilon}/S^\epsilon)$ using $O(S(1+\ln(T/S)))$ qubits.
\end{corollary}

In principle, it is always possible to turn a classical computation into a quantum one (\cref{cor:irreversible_to_quantum}) and combine all quantum algorithms into one quantum circuit by postponing all measurements until the very end of the computation, using the so-called \emph{principle of deferred measurement}~\cite{NC11}. We will use this fact implicitly in the rest of this work and just assume that we can take any classical algorithm and turn it into a quantum one with the same complexity.

\subsubsection{Quantum Search.} One of the most well-known quantum algorithms is Grover's unstructured search algorithm~\cite{Grover96}. Suppose we have a set of objects named $\{0,1,\dots, N-1\}$, of which some are \emph{targets}. We say that an oracle $\mathcal{O}$ \emph{identifies the targets} if, in the classical (resp.~quantum) setting, $\mathcal{O}(i)=1$ (resp. $\mathcal{O}\ket{i} = -\ket{i}$) when $i$ is a target and $\mathcal{O}(i)=0$ (resp. $\mathcal{O}\ket{i}= \ket{i}$) otherwise. Given such an oracle $\mathcal{O}$, the goal is to find a target $j \in \{0,1,\dots, N-1\}$ by making queries to the oracle $\mathcal{O}$.


In the search problem, we try to minimise the number of queries to the oracle. In the classical case, we need $O(N)$ queries to solve such a problem. Grover, on the other hand, provided a quantum algorithm that solves the search problem with only $O(\sqrt{N})$ queries~\cite{Grover96} when there is one target, and $O(\sqrt{N/t})$ when there are exactly $t$ targets. We here present a generalisation of Grover's algorithm called amplitude amplification~\cite{BHMT02}.

\begin{theorem}[Amplitude Amplification~\cite{BHMT02}]\label{lemma:grover}
Suppose we have a set of $N$ objects of which some are targets.
Let $\mathcal{O}$ be a quantum oracle that identifies the targets.
Let $\mathcal{A}$ be a quantum circuit using no intermediate measurements, ie $\mathcal{A}$ is reversible.
Let $a$ be the initial success probability of $\mathcal{A}$, that is the probability that a measurement of
$\mathcal{A}\ket0$ outputs a target.
There exists a quantum algorithm that calls $\bigO{\sqrt{1/a}}$ times $\mathcal{A}$, $\mathcal{A}^\dagger$
and $\mathcal{O}$, uses as many qubits as $\mathcal{A}$ and $\mathcal{O}$,
and outputs a target with probability greater than $1-a$.
\end{theorem}

Grover's algorithm is a particular case of this theorem where $\mathcal{A}$ produces a uniform superposition of all objects, in which case $a=\tfrac{1}{N}$. The theorem then states that we can find a target with probability $1-\tfrac{1}{N}$ using $O(\sqrt{N})$ calls to the oracle $\mathcal{O}_f$.




\begin{theorem}[Amplitude Estimation~\cite{BHMT02}, Theorem 12]\label{thm:amplitude_estimation}
Given a natural number $M$ and access to an $(n + 1)$-qubit unitary $U$ satisfying
\[
U\ket{0^n}\ket{0}= \sqrt{a}\ket{\phi_1}\ket{1} +\sqrt{1-a}\ket{\phi_0}\ket{0},
\]
where $\ket{\phi_1}$ and $\ket{\phi_0}$ are arbitrary $n$-qubit states and $0 < a < 1$,
there exists a quantum algorithm that uses $M$ applications of $U$ and $U^\dagger$, and outputs an estimate $\tilde{a}$ that with probability $\geq 2/3$ satisfies
\[
|a-\tilde{a}| \leq \frac{6\pi\sqrt{a(1-a)}}{M}+\frac{9\pi^2}{M^2}\leq \frac{15 \pi^2}{M}.
\]
\end{theorem}

We will have to search for a marked element in a collection but the oracle that identifies the targets may be probabilistic and return a wrong result with small probability. The following result generalises Grover search
in this setting. We say that a (probabilistic) Boolean function $f$ has
bounded error if there exists $b\in\{0,1\}$ such that $f()$ returns $b$
with probability at least $9/10$.

\begin{theorem}[\cite{Hoyer03}]\label{thm:hoyer03}
  Given $n$ algorithms, quantum or classical, each computing some bit-value with bounded error probability, there is a quantum algorithm that uses $O(\sqrt{n})$ queries and with constant probability: returns the index of a ``1'', if there is at least one ``1'' among the $n$ values; returns $\bot$ if there is no ``1''.
\end{theorem}

This algorithm can easily be used to find the index of the \emph{first algorithm} that returns $1$, see e.g.~\cite{KKM0Y21}.

\begin{lemma}\label{lem:quantum_find_first}
   Let $N$ be an integer and $f:[0,N-1] \to\bin$ a function.
   Let $\mathcal{O}$ be a (classical or quantum with bounded error) algorithm
   computing $f$.
   Let $n_0$ be the first index such that $f(n_0)=1$, or let $n_0=\bot$ if no such index exists.
   There exists a quantum algorithm $\mathcal{A}^\mathcal{O}$ with the following property.
   $\mathcal{A}^\mathcal{O}(N)$ returns $i\in[0,N-1]$ such that $f(i)=1$, or $\bot$.
   With constant probability, $\mathcal{A}^\mathcal{O}(N)=n_0$.
   The algorithm runs in expected time $T=O(\sqrt{n_0})$ (or $O(\sqrt{N})$ if $n_0=\bot$),
   uses a polynomial number of qubits and makes an expected number $T$ of calls to $\mathcal{O}$.
   Furthermore, if the algorithm returns $i\in[0,N-1]$, then it only queries $\mathcal{O}$ on
   values in $[0,\min(N-1,2i)]$.
\end{lemma}

\subsubsection{Memory Access.}\label{quantum_memory}

``Baseline'' quantum circuits are simply built using a universal quantum gate set. A requirement for many quantum algorithms to process data efficiently is to be able to access classical data in quantum superposition. Such algorithms use quantum random-access memory, often denoted as qRAM, and require the circuit model to be augmented with the so-called ``qRAM gate''. These qRAM gates are assumed to have a time complexity polylogarithmic in the amount of classical data stored, so that each call is not time consuming. This model is inspired by the classical RAM model where we usually assume memory access in time $O(1)$.

Given an input integer \(0 \leq i \leq r-1\), which represents the index of a memory cell, and many quantum registers \(\ket{x_0 , \ldots x_{r-1}}\), which represent stored data, the qRAM gate fetches the data from register \(x_i\), possibly in superposition:
\[ \ket{i} \ket{x_{0} , \ldots x_{r-1}} \ket{y} \mapsto \ket{i} \ket{x_0 , \ldots x_{r-1}} \ket{y \oplus x_i} \enspace.\]
Following the terminology of~\cite{DBLP:conf/tqc/Kuperberg13}, there are three types of qRAMs:
\begin{itemize}
\item If the input $i$ is classical, then this is the plain quantum circuit model. We can implement it using a universal quantum gate set.
\item If the $x_j$ are classical, we have \emph{classical memory with quantum random access } (QRACM).\label{acronym:qracm}
The qRAM gate becomes \[ \ket{i} \ket{y} \mapsto \ket{i} \ket{y \oplus x_i} \enspace.\]
\item In general, we have \emph{quantum memory with quantum random access } (QRAQM). This is the most powerful quantum memory model where the data are also in superposition.\label{acronym:qraqm}
\end{itemize}

In our algorithm for the dual attack, we will be using QRACM\@. It is possible to implement a QRACM using a universal quantum gate set, albeit at a considerable cost. Given a classical data set $\{x_0,\cdots,x_{r-1}\}$, one can construct, in time $\tilde{O}(r)$, a circuit using $\tilde{O}(r)$ qubits that implements a QRACM for this data set. The obtained circuit then allows query in the form $\ket{i} \ket{y} \mapsto \ket{i} \ket{y \oplus x_i} $ and has circuit depth $O(\polylog(r))$~\cite{GLM08,KP20,MGM20,HLGJ20}.
Note that even low depth implementation of QRACM has at least $\Omega(r)$ gates, hence has time complexity at least $\Omega(r)$ by our definition. Therefore, the assumption that the qRAM gates have time complexity $\polylog(r)$ is very strong and corresponds to parallel evaluation of the circuit. The feasibility of constructing efficient QRACM is further discussed in e.g.~\cite{AC:AGPS20}.

\subsection{The Classical Algorithm of~\cite{AC:GuoJoh21,Matzov22}}

In this section, we give an overview of the algorithm in~\cite{Matzov22}, reproduced in~\cref{alg:matzov}. Our quantum algorithm will be a modified version that relies essentially on the same analysis for the correctness but a new analysis for the quantum complexity. We give an overview of the involved parameters in \cref{tab:matzov-parameters}.

\begin{table}[htbp]
  \centering
  \caption{Dual attack parameters.}\label{tab:matzov-parameters}
  \begin{tabular}{l@{\hskip 1em}p{0.75\linewidth}}
    \toprule
    parameters & explanation\\
    \midrule
    \(n\), \(m\), \(\chi_{s}\), \(\chi_{e}\) & LWE parameters as in \cref{def:lwe}\\
    \(\beta_{0}, \beta_{1}\) & BKZ block size \(\beta_{0}\) and sieving dimension \(\beta_{1}\)\\
    \(p\) & modulus switching target modulus\\
    \(\mu\) & the target success probability \(0 < \mu < 1\)\\
    \(\sigma_{s}^{2}\), \(\sigma_{e}^{2}\) & variances of \(\chi_{s}\) and \(\chi_{e}\) respectively\\
    \(\vec{s}_{\enum}, \vec{s}_{\fft}, \vec{s}_{\rest}\) & components of \(\vec{s}\) covered by exhaustive search, FFT and lattice reduction\\
    \(\tilde{\vec{s}}_{\enum}, \tilde{\vec{s}}_{\fft}\) & guesses for \(\vec{s}_{\enum} \bmod q\) and \(\vec{s}_{\fft} \bmod p\)\\
    \(N_{\enum}(\vec{s}_\enum)\) & number of guesses with \(\tilde{\vec{s}}_{\enum}\) with probabilities larger than the probability of \(\vec{s}_{\enum}\) when drawing \(\vec{s}\) from \(\chi_{s}\).\\
    \(k_{\enum}, k_{\fft}, k_{\rest}\) & dimension of \(\vec{s}_{\enum}, \vec{s}_{\fft}\) or \(\vec{s}_{\rest}\) respectively\\
    \(\mat{A}_{\enum}, \mat{A}_{\fft},  \mat{A}_{\rest}\) & \(\mat{A}\cdot \vec{s}=\mat{A}_{\enum}\cdot \vec{s}_{\enum} + \mat{A}_{\fft} \cdot \vec{s}_{\fft} +\mat{A}_{\rest} \cdot \vec{s}_{\rest}\)\\
    \(\alpha\) & scaling/normalisation factor \(\alpha \coloneqq \sigma_{e}/\sigma_{s}\)\\
    \(L,D\) & list/number of short vectors returned by sieving oracle \(D \coloneqq |L|\)\\
    \(\psi\left(\vec{s}_{\fft}\right)\) & \(=\exp({\frac{2 \pi i c_{q^{\prime}}}{p} \sum_{t} s_{t}})\) where \(q' = q/\gcd(p,q)\), \(c_{q'} := 0\) when \(q'\) is odd and \(c_{q'} \coloneqq  1/2q'\) when \(q'\) is even\\
    \(C\) & FFT cutoff value for scoring function\\
    \(\Phi\) & defined in \Cref{sec:discrete_gaussian}\\
    \(\phi_{\mathrm{fp}}(\mu)\), \(\phi_{\mathrm{fn}}(\mu)\) & \(=\Phi^{-1}\left(1-\frac{\mu}{2 \cdot N_{\enum }\left(\vec{s}_{\enum}\right) \cdot p^{k_\fft}}\right)\), \(\Phi^{-1}\left(1-\frac{\mu}{2}\right)\)\\
     \(D_{\text{eq}}\) & exponential factor coming from the clean Fourier coefficient of the error in the dual attack equations for required samples\\
    \(D_{\text{round}}\) & exponential factor coming from the rounding when modulus switching for required samples\\
    \(D_{\text{arg}}\) & \(\approx 1/2\), improvement factor coming from considering the complex argument of the Fourier coefficient rather than only the magnitude.\\
    \(D_{\text{fpfn}}\) & a polynomial factor controlling false positives and negatives\\
    \bottomrule
  \end{tabular}

\end{table}

We are given a sample from $\LWE_{n,m,\chi_s,\chi_e}$, where $\chi_s$ and $\chi_e$ have small variance $\sigma_s^2$ and $\sigma_e^2$ respectively. We partition $\vec{s}$ into three components:
\[
\vec{s}=\begin{pmatrix}
	\vec{s}_{\enum}\\\vec{s}_{\fft}\\\vec{s}_{\rest}
\end{pmatrix}
\]
where $\vec{s}_{\enum}$ has $k_{\enum}$ coordinates, $\vec{s}_{\fft}$ has $k_{\fft}$ coordinates, and $\vec{s}_{\rest}$ has $k_{\rest}=n-k_{\enum}-k_{\fft}$ coordinates.  We split $\mat{A}$ into three components accordingly as well:
\[
   \mat{A} = \left[\mat{A}_{\enum} \ \mat{A}_{\fft} \  \mat{A}_{\rest}\right]
\]
so that \(\mat{A}\cdot \vec{s}=\mat{A}_{\enum}\cdot \vec{s}_{\enum} + \mat{A}_{\fft} \cdot \vec{s}_{\fft} +\mat{A}_{\rest} \cdot \vec{s}_{\rest}\). We define the matrix:
\[
\mat{B}=\begin{pmatrix}
\alpha\mat{I}_m & 0\\
\mat{A}_{\rest}^T & q\mat{I}_{k_{\rest}}
\end{pmatrix},
\]
where \(\alpha\) is a constant equal to \(\frac{\sigma_e}{\sigma_s}\) and is used for normalisation in the case that \(\vec{s},\vec{e}\) have different distributions. We find \(D\) short vectors of the form \(\begin{pmatrix}
\alpha \cdot \vec{x}_j\\\vec{y}_{j,\rest}
\end{pmatrix}\) in the column space of \(\mat{B}\) using some short vectors sampling procedure (see~\cref{alg:short_vectors_sampling}). Then, given a list \(L\) of \(D\) such vectors let \(\vec{y}_{j,\fft}\coloneqq \vec{x}_j^T\cdot \mat{A}_\fft\)
and \(\vec{y}_{j,\enum}\coloneqq \vec{x}_j^T\cdot \mat{A}_\enum\). We can then define the function \(F_L(\tilde{\vec{s}}_\enum,\tilde{\vec{s}}_\fft) =\)
{\small\[
    \Re\left(
        \frac{1}{\psi(\tilde{\vec{s}}_\fft)}\sum_{j}\exp\left(
            \left(\round{\frac{p}{q}\cdot\vec{y}_{j,\fft}}^T\cdot \tilde{\vec{s}}_\fft
            +\frac{p}{q}\cdot \vec{y}_{j,\enum}^T\cdot \tilde{\vec{s}}_\enum-\frac{p}{q}\cdot\vec{x}_j^T\cdot \vec{b}
            \right)\cdot\frac{2i\pi}{p}
        \right)
        \right)
\]}%
for all \(\tilde{\vec{s}}_\enum\in \ZZ_q^{k_\enum}\)
and \(\tilde{\vec{s}}_\fft\in\ZZ_p^{k_\fft}\). First, note that \(\tilde{\vec{s}}_\fft\) is mod \(p \ll q\) rather than \(q\) meaning that only up to \(p^{k_\fft}\) candidates have to be considered rather than \(q^{k_\fft}\). Applying modulus switching here is the key innovation of~\cite{Matzov22}, but since its details do not matter for us here, we refer the reader to~\cite{Matzov22} for details. Second, here \(\psi(\tilde{\vec{s}}_\fft)\) is a complex factor of norm \(1\) defined in~\cite[p.~25, proof of Lemma~5.4]{Matzov22} and easily computable (cf.~\cref{tab:matzov-parameters}). The function \(F_L\) essentially performs an FFT on values drawn
from a certain distribution.

Via an analysis that we do not reproduce, one can show that the function \(F_L\) above has the following properties with respect to some cutoff parameter \(C\) (computed below):
\begin{itemize}
\item If \(\tilde{\vec{s}}_\enum\neq \vec{s}_\enum\) then \(F_L(\tilde{\vec{s}}_\enum,\tilde{\vec{s}}_\fft)<C\) for all \(\tilde{\vec{s}}_\fft\in\ZZ_p^{k_\fft}\).
\item \(F_L(\vec{s}_\enum,\vec{s}_\fft)>C\)
\item There might be \(\tilde{\vec{s}}_\fft\neq\vec{s}_\fft\) such that
  \(F_L(\vec{s}_\enum,\tilde{\vec{s}}_\fft)>C\).
\end{itemize}
The first point corresponds to a wrong guess. In this case, values on which the FFT
is performed follow a uniform distribution and the expected value of $F_L(\tilde{\vec{s}}_\enum,\tilde{\vec{s}}_\fft)$ is $0$.
The second point corresponds to the correct guess. In this case, values on which the FFT
is performed essentially follow a normal distribution with nonzero mean and therefore the expected
value of $F_L(\vec{s}_\enum,\vec{s}_\fft)$ is nonzero. By carefully choosing the value of $C$,
and taking sufficiently many samples in the list, we can ensure that these properties hold
with high probability. The third point follows from the fact that~\cite{Matzov22}
performs a modulo switching operation that can introduce some errors and makes the analysis
of $F_L(\vec{s}_\enum,\tilde{\vec{s}}_\fft)$ with $\tilde{\vec{s}}_\fft\neq\vec{s}_\fft$ more
difficult. Consequently, it is simpler to assume that we can only recover $\vec{s}_\enum$
with certainty. We can therefore reformulate the algorithm of~\cite{Matzov22} as looking
for $\tilde{\vec{s}}_\enum$ such that there exists $\tilde{\vec{s}}_\fft$ such that
$F_L(\tilde{\vec{s}}_\enum,\tilde{\vec{s}}_\fft)>C$.

The claimed relationships hold with high probability over the choice of the elements in \(L\), assuming sufficiently many vectors. Here ``sufficiently many'' depends on the input LWE parameters as well as parameters of the dual attack algorithm. In particular, this magnitude depends on
\begin{enumerate*}[label=(\alph*)]
\item \(D_{\text{eq}}\) which is an exponential factor coming from the LWE error,
\item \(D_{\text{round}}\) which is an exponential factor coming from the rounding when modulus switching and
\item \(D_{\text{fpfn}}\) is a polynomial factor controlling false positives and negatives.
\end{enumerate*} The following lemma formally states the required magnitudes.
\begin{lemma}[Adapted from Theorem~5.2 in~\cite{Matzov22}]\label{lem:matzov-c-d}
  Let $(n,m,q,\chi_s,\chi_e)$ be LWE parameters and \((\beta_0,\beta_1,k_\enum,k_\fft,k_\rest,p)\) be parameters for \Cref{alg:matzov}. Let \(0 < \mu < 1\) be the targeted failure probability. Let \(\sigma_{e}\) be the standard deviation of \(\chi_{e}\), \(\sigma_{s}\) be the standard deviation of \(\chi_{s}\) and \(\alpha = \sigma_{e}/\sigma_{s}\). Denote by \(\ell\) the expected Euclidean length of the vectors returned by \Cref{alg:short_vectors_sampling}. Then, \Cref{alg:matzov} succeeds with probability at least \(1-\mu\) for
  \[
    C = \phi_{\mathrm{fp}}(\mu) \cdot \sqrt{D_{\mathrm{arg}} \cdot D} \quad \textnormal{and}\quad D \geq D_{\mathrm{eq}} \cdot D_{\mathrm{round}} \cdot D_{\mathrm{arg}} \cdot D_{\mathrm{fpfn}}(\mu)\]
 where
 \begin{align*}
   \phi_{\mathrm{fp}}(\mu),\ \phi_{\mathrm{fn}}(\mu) &= \Phi^{-1}\left(1-\frac{\mu}{2 \cdot N_{\enum }\left(\vec{s}_{\enum}\right) \cdot p^{k_\fft}}\right),\ \Phi^{-1}\left(1-\frac{\mu}{2}\right)\\
    D_{\mathrm{eq}} &=  \exp\left({4{\left(\frac{\pi \tau}{q}\right)}^{2}}\right) \textnormal{ for } \tau^{2}=\frac{\alpha^{-2}\cdot\|\vec{e}\|^{2}+\left\|\vec{s}_{\mathrm{lat}}\right\|^{2}}{m+k_{\mathrm{lat}}} \ell^{2},\\
    D_{\mathrm{round}}=& {\left(\prod_{\substack{t=0 \\ s_{t} \neq 0}}^{k_{{\mathrm{fft}}}-1}\left(\frac{\sin \left(\frac{\pi s_{t}}{p}\right)}{\frac{\pi s_{t}}{p}}\right)\right)}^{-2} \textnormal{ for } \vec{s}_{{\mathrm{fft}}} = (s_0, \ldots, s_{k_{{\mathrm{fft}}}-1}),\\
  D_{\arg} =& \frac{1}{2}+\exp\left({-8{\left(\frac{\pi \tau}{q}\right)}^{2}}\right) \textnormal{ and } D_{\mathrm{fpfn}}(\mu) = {\left(\phi_{\mathrm{fp}}(\mu) + \phi_{\mathrm{fn}}(\mu)\right)}^{2}.
\end{align*}
\end{lemma}

\begin{remark}
  In~\cite{Matzov22} two theorems are given establishing costs: Theorem~5.2 (essentially reproduced above), which is with respect to a fixed tuple \((\vec{s},\vec{e})\) and Theorem~5.9 which is with respect to distributions \(\chi_{s}, \chi_{e}\) and which will be restated in \Cref{lem:correctness_quantum_all}.
\end{remark}

\begin{algorithm}[h]
  \KwIn{LWE parameters $(n,m,q,\chi_s,\chi_e)$,
    integers $\beta_0,\beta_1\leqslant d$,
    integers $k_\enum,k_\fft,k_\rest$ such that
    $k_\enum+k_\fft+k_\rest=n$,
    an integer $p\leqslant q$,
    an integer $D$,
    a real number $C$,
    and an LWE pair $(\mat{A},\vec{b})\in\ZZ_q^{m\times n}\times \ZZ_q^m$.
  }
  \KwOut{(Guess of) the first $k_\enum$ coordinates of $\vec{s}$ or $\bot$.}
  Decompose $\mat{A}$ as $\begin{bmatrix}\mat{A}_\enum&\mat{A}_\fft&\mat{A}_\rest\end{bmatrix}$ of respective dimensions $m\times k_\enum$, $m\times k_\fft$ and $m\times k_\rest$.\;
  Compute the matrix $\mat{B}=\begin{bmatrix}\alpha \mat{I}_m&\mat{0}\\ \mat{A}_\rest^T&q\mat{I}_{k_\rest}\end{bmatrix}$
  where $\alpha=\tfrac{\sigma_e}{\sigma_s}$ \;
  Run \cref{alg:short_vectors_sampling} on the basis $\mat{B}$ with parameters $\beta_0,\beta_1,D$ to get a list $L$ of $D$ short vectors.
  \;\label{alg:matzov:line:bkz}
  \For{every value $\tilde{\vec{s}}_\enum$ in decreasing order of likelihood according to the
    secret distribution\label{alg:matzov:line:decreasing}}{
    Initialise a table $T$ of dimensions
    $\underbrace{p\times p\times \cdots p}_{k_\fft\text{ times}}$ \;
    \label{alg:matzov:line:loop_body_start}
    \For{every short vector $(\alpha \vec{x}_j,\vec{y}_\rest)$ in $L$}{
        \label{alg:matzov:line:fft_fill_start}
      Compute $\vec{y}_{j,\fft}=\vec{x}_j^T\cdot \mat{A}_\fft$. \;
      Compute $\vec{y}_{j,\enum}=\vec{x}_j^T\cdot \mat{A}_\enum$. \;
      Add $\exp\left((\vec{x}_j^T \cdot \vec{b} -\vec{y}^T_{j,\enum}\cdot \tilde{\vec{s}}_\enum)\cdot \tfrac{2i\pi}{q}\right)$
      to cell $\lfloor\tfrac{p}{q}\vec{y}_{j,\fft}\rceil$ of $T$. \;\label{alg:matzov:line:open-problem}
    }
    \label{alg:matzov:line:fft_fill_end}
    Perform FFT on $T$ \;
    \label{alg:matzov:line:do_fft}
    \If{for any $\tilde{\vec{s}}_\fft$, the real part of
      $\tfrac{1}{\psi(\tilde{\vec{s}}_\fft)}T[\tilde{\vec{s}}_\fft]$ is larger than $C$}{
      \label{alg:matzov:line:fft_test}
      \Return{$\tilde{\vec{s}}_\enum$.}
    }
    \label{alg:matzov:line:loop_body_end}
  }
  \Return{$\bot$}
  \caption{\label{alg:matzov}Dual Attack of~\cite{Matzov22}}
\end{algorithm}

\begin{remark}
In \cref{alg:matzov:line:decreasing} the algorithm asks to enumerate candidates for $\tilde{\vec{s}}_\enum$ in decreasing order of likelihood. A technique for achieving this enumeration is discussed in~\cite[Section~5.4]{Matzov22}. Roughly, this is accomplished by enumerating with increasing bounds on log-likelihoods for candidates, see~\cite{Matzov22} for details.
\end{remark}

\section{Quantum Guessing}\label{sec:guessing}

In this section, we give a more rigorous analysis of the classical and quantum guessing complexity for targets following a discrete Gaussian distribution. In more detail, let $X$ be a random variable on a finite or countable set. We consider the problem of guessing the value taken by $X$ by asking questions of the form ``Is $X$ equal to $x$?'' until the answer is yes. This problem arises when we must find the secret $\vec{s}_\enum$ in the dual attack by asking the question ``is the secret equal to $\tilde{\vec{s}}_\enum$?''. Let $N$ be the number of guesses used in the guessing strategy that minimises $\Exp{}{N}$. It can be shown that the best strategy is to try values of $X$ in decreasing order of probability. Without loss of generality, we can identify the possible values of $X$ with $\NN$ in such a way that $p_0\geq p_1\geq p_2\geq \cdots $ where $p_i=\Prob{}{X=i}$. The expected number of guesses of the optimal strategy is therefore
\[
G(X)=\sum_ii\cdot p_i.
\]
It is well-known that a \emph{lower bound} on $G(X)$ is given by the entropy of $X$.
More precisely, Massey showed in~\cite{Massey94} that
\[
G(X)\geq \tfrac{1}{4}\cdot 2^{H(X)}+1
\]
provided that $H(X)\geq 2$ bits, where $H$ denotes Shannon's entropy (i.e.~in base 2).
On the other hand, the same work shows that it is not, in general, possible to bound $G(X)$
in terms of $H(X)$ only. In \Cref{lem:guessing_complexity_of_discrete_gaussian_multidim}, we
heuristically show that
$G(X)\approx {(\tfrac{2}{\sqrt{e}})}^n\cdot 2^{H(X)}$
when $X$ is distributed according to a $n$-dimensional discrete Gaussian.

In this work, we are interested in the \emph{quantum complexity} of guessing.
Montanaro showed~\cite{Montanaro11} that the expected number of guesses in this case becomes,
\[
     G^{qc}(X)=\sum_i\sqrt{i}\cdot p_i
\]
and that this is the best possible. However,~\cite{Montanaro11} only deals with the case where the oracle (which answers the question ``is $X$ equal to $x$?'') always returns the correct answer. Using \Cref{lem:quantum_find_first}, a variant of Grover search that can handle two-sided errors, we extend this algorithm to deal with bounded error oracles.

\begin{lemma}\label{lem:quantum_guessing}
  Let $X$ be a random variable taking values in some (effectively describable) set $E$. Assume that there is an efficiently computable bijective function $\sigma:\NN\to E$ such that for all $i\leqslant j$, $\Prob{X}{X=\sigma(i)}\geqslant\Prob{X}{X=\sigma(j)}$, \emph{i.e.} $\sigma$ orders $E$ by non-increasing probability according to $X$. Let ${(\mathcal{O}_x)}_{x\in E}$ be a collection of oracles such that for any $x\in E$, $\mathcal{O}_x(x)$ returns $1$ with probability at least $9/10$ and for all $y\neq x$, $\mathcal{O}_x(y)$ returns $0$ with probability at least $9/10$. Then there is a quantum algorithm $\mathcal{A}$, with quantum oracle access to $\sigma$ and $\mathcal{O}_x$ such that for all $x\in E$, $\mathcal{A}^{\sigma,\mathcal{O}_x}()=x$ with constant probability, and
    \[
        \Exp{X}{T(X)}=O(G^{qc}(X)),
        \qquad
        \Exp{X}{Q(X)}=O(G^{qc}(X)),
    \]
    where $T(x)$ is the running time complexity of
    $\mathcal{A}^{\sigma,\mathcal{O}_x}()$, and $Q(x)$ its query
    complexity.
\end{lemma}

\begin{proof}
    See \Cref{sec:proof_guessing_complexity}.
\end{proof}

We now study the guessing complexity of a \(n\)-dimensional discrete Gaussian, its modular version and also relate these quantities to the entropy. The reason why we also study the (non-modular) Gaussian is that it is not clear how to order the elements by decreasing probability in the modular case, whereas it is easy in the non-modular one. Therefore we study the discrete Gaussian first and then show that its guessing complexity is an upper bound on the guessing complexity of the modular Gaussian.

\begin{lemma}\label{lem:guessing_complexity_of_discrete_gaussian_multidim}
    Let $n\geqslant 4$. Then
    \[
        1\leqslant \frac{\rho_\sigma(\ZZ^n)}{(\sigma\sqrt{2\pi})^n}\leqslant \coth(\pi^2\sigma^2),
        \qquad
        G(D_{\ZZ^n,\sigma})\leqslant
	        \frac{2^n\rho_\sigma(\ZZ^n)}{1-e^{-1/2\sigma^2}},
	\]
	and
	\[
	    -n\frac{8\pi^2\sigma^2
            e^{-2\pi^2\sigma^2}}{\log(2)\coth(\pi^2\sigma^2)}
	    \leqslant
	   H(D_{\ZZ^n,\sigma})-\frac{\tfrac{1}{2}+\log(\sigma\sqrt{2\pi})}{\log(2)}n
	    \leqslant
            n\log_2(\coth(\pi^2\sigma^2)).
	\]
	Furthermore, if $\sigma\geqslant \sqrt[4]{\frac{2}{27\pi^2}}\approx 0.294$ then
	\[
	    G^{qc}(D_{\ZZ^n,\sigma})\leqslant
	        \frac{7}{6}\cdot\left(\frac{3}{2}\right)^{3n/4}
            \frac{\sqrt{\rho_{\sigma}(\ZZ^n)}}{(1-e^{-1/3\sigma^2})^{3/2}}.
	\]
	Furthermore,
	\begin{align*}
	    G(D_{\ZZ^n,\sigma})
	    &\leqslant \frac{1}{1-e^{-1/2\sigma^2}}\left(\frac{2a(\sigma)}{\sqrt{e}}\right)^n
            \cdot 2^{H(D_{\ZZ^n,\sigma})},\\
        G^{qc}(D_{\ZZ^n,\sigma})
	    &\leqslant \frac{7}{6}\cdot\frac{1}{(1-e^{-1/3\sigma^2})^{3/2}}
	    \left(\frac{27a(\sigma)^2}{8e}\right)^{n/4}
            \cdot 2^{H(D_{\ZZ^n,\sigma})/2}
	\end{align*}
	where $a(\sigma)=e^{8\pi^2\sigma^2
            e^{-2\pi^2\sigma^2}\tanh(\pi^2\sigma^2)}$
    is plotted on Figure~\ref{fig:alpha_sigma} and
    \[
        a(\sigma)=1+8\pi^2\sigma^2 e^{-2\pi^2\sigma^2}+o(\sigma^2 e^{-\pi^2\sigma^2}),
        \qquad
        \sigma\to\infty.
    \]
	Finally, for all $n$ and $\sigma$,
	\[
	    G(D_{\ZZ_q^n,\sigma})\leqslant 2\,G(D_{\ZZ^n,\sigma}),
	    \qquad
	    G^{qc}(D_{\ZZ_q^n,\sigma})\leqslant \frac{3}{2}\,G^{qc}(D_{\ZZ^n,\sigma}).
	\]
\end{lemma}
\begin{proof}
	See \Cref{sec:proof:guessing_complexity_of_discrete_gaussian_multidim}.
\end{proof}

In \Cref{lem:guessing_complexity_of_discrete_gaussian_multidim}, we have identified upper bounds which we believe to be tight up to polynomial factors in the dimension. It is not clear that the same method can be used to obtain lower bounds. The bounds we obtained suggest (but do not conclusively show) that the quantum guessing complexity is exponentially smaller than the square root of the classical guessing complexity, for the discrete Gaussian over $\ZZ^n$. Obtaining lower bounds on these quantities would confirm that this is indeed the case. The situation is the same for the \emph{modular} discrete Gaussian since we also only obtained upper bounds for this distribution.


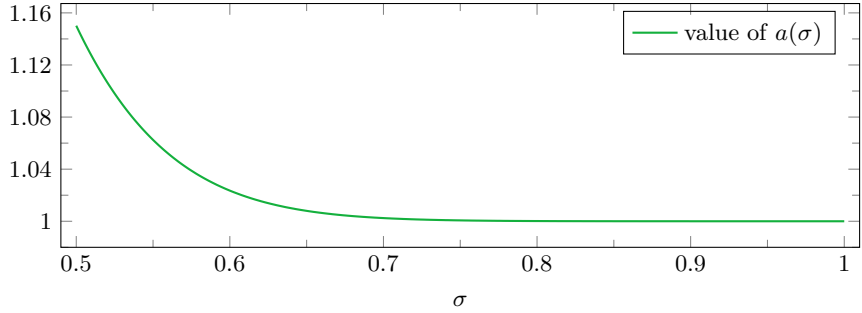
\begin{figure}
    \centering
    \begin{tikzpicture}[scale = 1]
        \begin{axis}[xlabel={$\sigma$},ylabel={},
                width=\textwidth,
                height=0.25\textheight,
                minor y tick num=1,
                xtick distance=0.1,
                ytick distance=0.04,
                minor x tick num=1,
                enlarge x limits=0.02,
                legend style={legend pos=north east},
                legend cell align={left},
                domain=0.5:1,
                ymin=0.98
                ]
            \addplot+[samples=400] 
            {exp(8*pi*pi*x*x*exp(-2*pi*pi*x*x)*tanh(pi*pi*x*x))};
            \addlegendentry{value of $a(\sigma)$};
        \end{axis}
    \end{tikzpicture}
    \caption{Value of the extra factor $a(\alpha)$ in \Cref{lem:guessing_complexity_of_discrete_gaussian_multidim}.}\label{fig:alpha_sigma}
\end{figure}

\section{A Quantum Algorithm for Mean Estimation}\label{sec:quantum-mean-estimation}

We provide a quantum algorithm which estimates the mean value of \[\cos(2\pi(\langle\vec{w}_i,\vec{b}\rangle)/q)\] used in the dual attack. The idea is inspired by~\cite[Theorem~47]{ACKS20} and can be seen as a special case of quantum speedup of Monte Carlo methods~\cite{Montanaro15}.

\begin{definition}[QRACM Oracle]\label{def:qracm-oracle}
  Let \(N\) be a positive integer and \(W\) be a list of \(N\) vectors \(\vec{w}_0,\dots,\vec{w}_{N-1} \in \mathbb{Z}^n\). The QRACM oracle for \(W\) is
  \[
    \mathcal{O}_W: \ket{j}\ket{0}\mapsto \ket{j}\ket{\vec{w}_j}.
  \]
\end{definition}

\begin{definition}[Positive Controlled Rotation Oracle]\label{def:ocrp-oracle}
The positive controlled rotation oracle for any $a\in\mathbb{R}$ is
\[
\mathcal{O}_{CR^+}: \ket{a} \ket{0} \rightarrow \begin{dcases}
    \ket{a}(\sqrt{a}\ket{1}+\sqrt{1-a}\ket{0}),& \text{if } a\geq 0\\
    \ket{a}\ket{0},              & \text{otherwise,}
\end{dcases}
\]
which can be implemented up to negligible error by $\poly[\log n]$ quantum elementary gates.
\end{definition}

\begin{definition}[Cosine Inner Product Oracle]\label{def:cos-oracle}
 The cosine inner product oracle for any $\vec{b},\vec{w}\in \mathbb{Z}^n$ is
\[
  \mathcal{O}_{\cos}: \ket{\vec{w}}\ket{\vec{b}}\ket{0}\rightarrow
  \ket{\vec{w}}\ket{\vec{b}}\ket{\cos(2\pi \langle \vec{w},\vec{b}\rangle/q)},
\]
which can be implemented by $\poly[\log n]$ quantum elementary gates.
\end{definition}

\begin{theorem}\label{thm:Quantum_estimation_f_W}
  Let \(N\) be a positive integer and \(W\) be a list of \(N\) vectors \(\in \mathbb{Z}^n\): \(\vec{w}_0,\dots,\vec{w}_{N-1}\). Let \(f_{W}(\vec{b})=\tfrac{1}{N}\sum_{i=0}^{N-1} \cos(2\pi(\langle\vec{w}_i,\vec{b}\rangle)/q)\), where \(\Vec{b}\in \mathbb{Z}_q^n\). For any \(\epsilon,\delta>0\), there exists a quantum algorithm \(\mathcal{A}\) that given \(\vec{b}\in \mathbb{Z}_q^n\) and oracle access to \(\mathcal{O}_W\) as defined in \cref{def:qracm-oracle}, outputs \(\mathcal{A}^{\mathcal{O}_W}(\vec{b})\) which satisfies \(|\mathcal{A}^{\mathcal{O}_W}(\vec{b})-f_{W}(\vec{b})|\leq \epsilon\) with probability \(1-\delta\). The algorithm makes \(\mathcal{O}(\epsilon^{-1}\cdot\log\frac{1}{\delta})\) queries to \(\mathcal{O}_W\), and requires \(\epsilon^{-1}\cdot\log\frac{1}{\delta}\cdot \poly[\log n]\) elementary quantum gates.
\end{theorem}

\begin{proof}
Prepare the state $\frac{1}{\sqrt{N}}\sum\limits_{j=0}^{N-1}\ket{j}\ket{\vec{0}}\ket{\vec{b}}\ket{0}\ket{0}$, and then apply $\mathcal{O}_W$ on the first and second registers (storing \(\vec{w}_{j}\) there), apply $\mathcal{O}_{\cos}$ on the second, third, fourth registers (storing \(\cos(2\pi \langle \vec{w},\vec{b}\rangle/q)\) there), and apply $\mathcal{O}_{CR^+}$ on the fourth and fifth registers. Writing \(\gamma_{j} \coloneqq \cos(2\pi\langle \vec{w}_j,\vec{b}\rangle/q) \) and letting sums run over \(j \in [0,N-1]\), we have
\begin{align*}
  \frac{1}{\sqrt{N}}\sum\limits_{\substack{\gamma_j \ge 0}} \ket{j}\ket{\vec{w}_j}\ket{\vec{b}}\ket{\gamma_j} \big(\sqrt{\gamma_j}\ket{1} + \sqrt{1-\gamma_j}\ket{0}\big)
  + \frac{1}{\sqrt{N}}\sum\limits_{\substack{\gamma_{j} <0}} \ket{j}\ket{\vec{w}_j}\ket{\vec{b}}\ket{\gamma_j}\ket{0}.
\end{align*}
By rearranging, we obtain
\begin{align*}
  &\frac{1}{\sqrt{N}}
    \sum_{\substack{\gamma_j \ge 0}} \sqrt{\gamma_j}\ket{j}\ket{\vec{w}_j}\ket{\vec{b}}\ket{\gamma_j}\ket{1}\\
  +& \frac{1}{\sqrt{N}} \bigg(\sum_{\gamma_j\ge 0}\sqrt{1-\gamma_j}\ket{j}\ket{\vec{w}_j}\ket{\vec{b}}\ket{\gamma_j}
+\sum_{\gamma_{j}<0} \ket{j}\ket{\vec{w}_j}\ket{\vec{b}}\ket{\gamma_j}\bigg)\ket{0}\\
= & \sqrt{a^+}\ket{\phi_1}\ket{1}+\sqrt{1-a^+}\ket{\phi_0}\ket{0},
\end{align*}
where $a^+=\sum_{\gamma_j \ge 0} \frac{\gamma_j}{N}$. By applying Theorem~\ref{thm:amplitude_estimation}, we can estimate $a^+$ with additive error $\epsilon/2$ by using  $\mathcal{O}(\epsilon^{-1})$ applications of $\mathcal{O}_W$, $\mathcal{O}^\dagger_W$, and $\epsilon^{-1}\cdot \poly[\log n]$ elementary quantum gates. Following the same strategy, we can also estimate $a^-=\sum_{\gamma_j<0} \frac{\gamma_j}{N}$ with same additive error and by using same amount of queries and quantum elementary gates. Therefore, we can estimate \[a^++a^- \pm \epsilon = \sum_{j} \frac{\cos(2\pi\langle \vec{w}_j,\vec{b}\rangle/q)}{N}.\] By repeating the procedure $\Theta(\log \frac{1}{\delta})$ times and taking the median among them, we finish the proof. \qed{}
\end{proof}

\section{Quantum Augmented Dual Attack}\label{sec:quantum_dual_attack}

We now modify the algorithm of~\cite{Matzov22} to obtain a quantum speedup. At a high-level, \cref{alg:quantum_dual_attack} works in the same way. First, we run a sampling algorithm to obtain short vectors in the dual. Here, we can take advantage of the existing quantum speedups for sieving~\cite{PQCRYPTO:LaaMosvan13,PhD:Laarhoven15,AC:AGPS20,AC:ChaLoy21}.

Next, we can obtain an (at least) quadratic speedup on the search for $\vec{s}_\enum$.
The original algorithm enumerates them one by one until the correct one is found. By carefully choosing the order in which elements are enumerated, one can show that the expected complexity of this search is related to the guessing complexity  (see \Cref{sec:guessing}). In the quantum setting,
we can apply a variant of Grover's search
algorithm to obtain a speed-up on this search. The complexity
of this search is now related to the quantum guessing
complexity which is not necessarily related to the (classical)
guessing complexity. Indeed, this quantity is
always smaller than the square root of the classical one and our calculations suggest that it is strictly smaller than the square root for the discrete Gaussian.
In our case, the quantum search will call an oracle that is probabilistic so care must be taken. We use the improved version of Grover's search in \Cref{thm:hoyer03} that can handle bounded-error inputs.

We now move to the most interesting part of our quantum speedup. In their algorithm~\cite{Matzov22}, the authors first fill a large array T, perform an FFT and then look at all the entries to check if one is larger than a given threshold $C$. While it is tempting to use the quantum Fourier transform (QFT), which runs in polynomial time, we do not know how to implement the second step (checking each entry) efficiently. Indeed, the QFT works on the amplitudes and, therefore, simply extracting a coefficient of the result is a nontrivial task (see \Cref{sec:open_problem}). We work around this issue by observing two points:
\begin{enumerate}
\item The input array of the FFT is relatively sparse: it has $D$ nonzero entries (out of $p^{k_\fft}$).
\item We can obtain a quadratic speedup on the task of evaluating a sum of cosine (\Cref{thm:Quantum_estimation_f_W}).
\end{enumerate}
Since every entry of the output of the FFT is a sum of cosine that we can evaluate efficiently,
and since the sum only has $D$ terms, we can evaluate each coefficient in reasonable time.
By turning this algorithm into a quantum oracle, we can use Grover's search to obtain
a further quadratic speedup on the inner part of the algorithm that looks for an entry above the
threshold $C$.

A crucial detail of this algorithm is our use of a QRACM\@. Indeed, in order to apply \Cref{thm:Quantum_estimation_f_W} and obtain a quadratic speedup when evaluating the sums, we need a quantum oracle access to the short vectors stored in \(L\). Since those vectors are obtained by a classical algorithm, we store them in a QRACM to build this oracle.\footnote{Note that quantum augmented sieving procedures still output classical lists of short vectors.}


\begin{algorithm}[htbp]
    \SetKwProg{Oracle}{create oracle}{:}{}
    \KwIn{LWE parameters $(n,m,q,\chi_s,\chi_e)$,
        integers $\beta_0,\beta_1\leqslant d$,
        integers $k_\enum,k_\fft,k_\rest$ such that
        $k_\enum+k_\fft+k_\rest=n$,
        an integer $p\leqslant q$,
        an integer $D$,
        a real number $C$,
        a coefficient $\eta\in[0,1]$
        and an LWE pair $(\mat{A},\vec{b})\in \ZZ_q^{m\times n}\times \ZZ_q^m$.
        }
    \KwOut{(Guess of) the first $k_\enum$ coordinates of $\vec{s}$ or $\bot$.}
    Decompose $\mat{A}$ as
    $\begin{bmatrix}\mat{A}_\enum&\mat{A}_\fft&\mat{A}_\rest\end{bmatrix}$ of
    respective dimensions $m\times k_\enum$, $m\times k_\fft$ and
    $m\times k_\rest$.\; \label{alg:quantum_dual_attack:init_start}
    Compute the matrix $\mat{B}=
        \begin{bmatrix} q\mat{I}_{k_\rest} & \mat{A}_\rest^T\\
        \mat{0} & \alpha \mat{I}_m\\ \end{bmatrix}$
    where $\alpha=\tfrac{\sigma_e}{\sigma_s}$ \;
    Run \Cref{alg:short_vectors_sampling} on the basis $\mat{B}$ with parameters $\beta_0,\beta_1,D$ to get a list $L$ of $D$
        short vectors. \; \label{alg:line:bkz}
        \label{alg:quantum_dual_attack:init_end}
    Create a QRACM $\mathcal{O}_W$ \; \label{alg:quantum_dual_attack:create_qram}
    \For{every short vector $(\alpha \cdot\vec{x}_j,\vec{y}_{j,\rest})$ in $L$}{
        \label{alg:quantum_dual_attack:loop_qram}
        Add vector $\vec{x}_j$ to $\mathcal{O}_W$ at index $j$ \;
        \label{alg:quantum_dual_attack:store_qram}
    }
    Use \Cref{thm:Quantum_estimation_f_W} to create an algorithm $\mathcal{A}$
        with
        $\delta=\tfrac{1}{10}$,
        $\varepsilon=\tfrac{C}{D}\eta$ and ``$q$''=$p$ \;
        \label{alg:quantum_dual_attack:create_quant_est_alg}
    \Oracle{\textup{$\mathcal{O}$($\tilde{\vec{s}}_\enum$)}}{
        \label{alg:quantum_dual_attack:create_oracle_O}
        \Oracle{\textup{$\hat{\mathcal{O}}$($\tilde{\vec{s}}_\fft$)}}{
            Compute $\theta$ such that $\psi(\tilde{\vec{s}}_\fft)=e^{-\tfrac{2i\pi}{p}\theta}$
                (recall that $|\psi(\tilde{\vec{s}}_\fft)|=1$) \;
            \label{alg:quantum_dual_attack:compute_theta}
            \Oracle{\textup{$\mathcal{O}_W'$($j$)}}{
                \label{alg:quantum_dual_attack:create_oracle_O'}
                Get $\vec{x}_j$ from $\mathcal{O}_W$ at index $j$ \;
                \label{alg:quantum_dual_attack:get_xj}
                Compute $\vec{y}_{j,\fft}=\vec{x}_j^T\cdot \mat{A}_\fft$ \;
                Compute $\vec{y}_{j,\enum}=\vec{x}_j^T\cdot \vec{A}_\enum$ \;
                \label{alg:quantum_dual_attack:calc_yjenum}
                \Return{vector $\left(\tfrac{p}{q}\cdot \vec{y}_{j,\enum},
                    \round{\tfrac{p}{q}\cdot \vec{y}_{j,\fft}}, \theta-\tfrac{p}{q}\cdot\vec{x}_j^T\cdot \vec{b}\right)$}
            }
            \label{alg:quantum_dual_attack:end_oracle_O'}
            \eIf{$\mathcal{A}^{\mathcal{O}_W'}((\tilde{\vec{s}}_\enum, \tilde{\vec{s}}_\fft,1))
                    >(1+\eta)\cdot\tfrac{C}{D}$}{
                \label{alg:quantum_dual_attack:test_fft}
                \Return{1}
            }{
                \Return{0}
            }
        }
        Use \Cref{thm:hoyer03} to find, with probability $\tfrac{9}{10}$, $i$ such that $\hat{\mathcal{O}}(i)=1$
            or let $i=\bot$ if none exists \;
            \label{alg:quantum_dual_attack:grover_for_fft}
        \eIf{$i\neq \bot$}{
                \Return{1}
            }{
                \Return{0}
            }
    }
    \label{alg:quantum_dual_attack:end_oracle_O}
    \Oracle{\textup{$\tilde{\mathcal{O}}$($i$)}}{
        \label{alg:quantum_dual_attack:create_Otilde}
        Compute the $i^{th}$ most probable $\tilde{\vec{s}}_\enum$ according to
        the distribution $\chi_s^{k_\enum}$ \;
        \Return{$\mathcal{O}(\tilde{\vec{s}}_\enum)$}
    }
    Find, with probability $\tfrac{9}{10}$, $\tilde{\vec{s}}_\enum$ such that $\mathcal{O}(\tilde{\vec{s}}_\enum)=1$
        using \Cref{lem:quantum_guessing} 
        with oracle $\tilde{\mathcal{O}}$, or let $\tilde{\vec{s}}_\enum=\bot$
        if none is found\; \label{alg:quantum_dual_attack:guess}
    \Return{$\tilde{\vec{s}}_\enum$.}
    \caption{\label{alg:quantum_dual_attack}Quantum Augmented Dual Attack}
\end{algorithm}

\subsection{Correspondance between the classical and quantum algorithms}

Before delving into the analysis of \Cref{alg:quantum_dual_attack}, we start by explaining the major steps of the algorithm. This is best done by giving the correspondence between our algorithm, and the original algorithm from~\cite{Matzov22} (\Cref{alg:matzov}).
\begin{itemize}
\item Steps~\ref{alg:quantum_dual_attack:init_start}-%
  \ref{alg:quantum_dual_attack:init_end} are exactly the same for the two algorithms.
\item Steps~\ref{alg:quantum_dual_attack:create_qram}-%
  \ref{alg:quantum_dual_attack:store_qram} create and initialise the QRACM that contains all the short dual vectors. This step is not needed in the original algorithm since the vectors are simply stored in an array.
\item Step~\ref{alg:quantum_dual_attack:create_quant_est_alg} instantiates the quantum estimation algorithm with certain parameters. This step is only here for readability.
\item Steps~\ref{alg:quantum_dual_attack:create_oracle_O}-%
  \ref{alg:quantum_dual_attack:end_oracle_O} roughly correspond to Steps~\ref{alg:matzov:line:loop_body_start}-%
  \ref{alg:matzov:line:loop_body_end} of \Cref{alg:matzov}. Specifically, our algorithm will perform a quantum search using those steps as an oracle, whereas \Cref{alg:matzov} uses those steps as a body of a search loop. The result is the same in the sense that our oracle $\mathcal{O}$ essentially produces the same result as the body of the loop in \Cref{alg:matzov}. The details on how this result is obtained are different however:
  \begin{itemize}
  \item Step~\ref{alg:quantum_dual_attack:compute_theta} is used to compute the phase of the normalisation factor in the FFT sum. This factor is needed at Step~\ref{alg:matzov:line:fft_test} of \Cref{alg:matzov} but we need to put it in the search oracle because of the way the search algorithm works.
  \item Steps~\ref{alg:quantum_dual_attack:create_oracle_O'}-%
    \ref{alg:quantum_dual_attack:end_oracle_O'} roughly correspond to Steps-\ref{alg:matzov:line:fft_fill_start}-%
    \ref{alg:matzov:line:do_fft}. As explained at the beginning of \Cref{sec:quantum_dual_attack} and also in \Cref{sec:open_problem}, our quantum algorithm does not perform the FFT in the same way as the original algorithm. Contrary to the classical algorithm where all the output coefficients are computed ``at once'' by a single FFT computation, our algorithm computes each individual output FFT coefficient on demand and check whether it is above the threshold (this is what $\hat{\mathcal{O}}$ does). To do so, given a specified output coefficient $\tilde{s}_\fft$, we define an oracle $\mathcal{O}'_W$ which (essentially) returns the (input) coefficients of the FFT and use \Cref{thm:amplitude_estimation} to estimate the output coefficient. We note that the steps~%
    \ref{alg:quantum_dual_attack:get_xj}-%
    \ref{alg:quantum_dual_attack:calc_yjenum} could have been in the oracle $\mathcal{O}_W$ instead, this would be completely equivalent (and save a small polynomial factor in the overall complexity). Step~\ref{alg:quantum_dual_attack:test_fft} corresponds to Step~\ref{alg:matzov:line:fft_test} of \Cref{alg:matzov}. The difference between the thresholds comes from the way \Cref{thm:amplitude_estimation} works ($C$ vs $C/D$) and also because of the two-sided errors made by the oracle (factor $\eta$). See \Cref{th:analysis_quantum_alg} for more details.
  \end{itemize}
\item Steps~\ref{alg:quantum_dual_attack:create_Otilde}-%
  \ref{alg:quantum_dual_attack:guess} correspond to Step~\ref{alg:matzov:line:loop_body_start} of \Cref{alg:matzov}. Specifically, the classical algorithm performs a sequential search by decreasing probability of $\tilde{s}_\enum$. Our algorithm uses \Cref{lem:quantum_guessing} to obtain a quantum speedup on this search which requires to create an oracle.
\end{itemize}

\subsection{Analysis of the Quantum Augmented Dual Attack}

We now analyse the quantum augmented dual attack given in \Cref{alg:quantum_dual_attack}. Informally, we first establish that with constant probability the algorithm outputs a guess for \(\vec{s}_{\enum}\) where the corresponding FFT scoring function's score is above the chosen threshold \(C\). Below, we will instantiate this theorem with appropriate choices of thresholds \(C\) (and number of samples \(D\)). Note that like~\cite[Thm.~5.2]{Matzov22}, the following theorem is with respect to a fixed tuple \((\vec{s},\vec{e})\).

\begin{theorem}\label{th:analysis_quantum_alg}
    Let $(n,m,q,\chi_s,\chi_e)$ be LWE parameters, let \[(\beta_0,\beta_1,k_\enum,k_\fft,k_\rest,p,D,C,\mat{A},\vec{b},\eta)\] be the input of \Cref{alg:quantum_dual_attack}. Let $L$ be the list of vectors obtained at
    \Cref{alg:line:bkz} of \Cref{alg:quantum_dual_attack}.
    For any $x>0$, let $S_x^L=\set{\tilde{\vec{s}}_\enum:\exists \tilde{\vec{s}}_\fft, F_L(\tilde{\vec{s}}_\enum,\tilde{\vec{s}}_\fft)>x}$. With probability at least $9/10$,
    the algorithm returns a value in $S_{C}^L\cup\set{\bot}$. Furthermore,
    if $S_{(1+2\eta)C}^L\neq\varnothing$ then the algorithm returns a value
    in $S_{C}^L$ with probability at least $9/10$.
\end{theorem}
\begin{proof}
See \cref{sec:proofs-for-algorithm}.
\end{proof}

Next, we establish that the parameter choices of \(C,D\) in \cref{lem:matzov-c-d} for \cref{alg:matzov}  allow us to also instantiate \cref{alg:quantum_dual_attack} such that it succeeds with a comparable probability.

\begin{lemma}\label{lem:correctness_quantum}
  Let $(n,m,q,\chi_s,\chi_e)$ be $\LWE$ parameters, let $(\beta_0,\beta_1,k_\enum,k_\fft,k_\rest,p$, $C,D,\eta)$ be a tuple of parameters for \Cref{alg:quantum_dual_attack}, and let $0<\nu<1$. Fix $(\vec{s},\vec{e})\in \ZZ_q^n\times \ZZ_q^{m}$. By choosing the parameters $C,D$ according to \Cref{lem:matzov-c-d} with $\mu=\nu/2$, and $\eta\leqslant\frac{\sqrt{2\pi}\mu}{8\phi_{\text{fp}}(\mu)}$, \Cref{alg:quantum_dual_attack} returns $\vec{s}_\enum$ with probability at least $1-\nu$.
\end{lemma}
\begin{proof}
See \cref{sec:proofs-for-algorithm}.
\end{proof}

Next, we state the analogous result as~\cite[Thm.~5.9]{Matzov22} for our setting, i.e.~express the success probability with respect to \(\chi_{s},\chi_{e}\) rather than a fixed tuple \((\vec{s},\vec{e})\). The only difference compared to~\cite[Thm.~5.9]{Matzov22} is that we replace the Shannon entropy \(H(X)\) by the guessing complexity \(G(X)\), as discussed in~\cref{sec:guessing}.

\begin{lemma}\label{lem:correctness_quantum_all}
  Let $(n,m,q,\chi_s,\chi_e)$ be $\LWE$ parameters, $(\beta_0,\beta_1,k_\enum,k_\fft,k_\rest,p$, $D, C, \eta)$ be parameters for \Cref{alg:quantum_dual_attack}. Let $0<\nu<1$, $\mu=\nu/2$, $\eta\leqslant\frac{\sqrt{2\pi}\mu}{8\phi_{\text{fp}}(\mu)}$. Denote by \(\ell\) the expected Euclidean length of the vectors returned by \Cref{alg:short_vectors_sampling}. Let \(G(X)\) as in \cref{lem:guessing_complexity_of_discrete_gaussian_multidim}. Let
  \(D \geq D_{\mathrm{eq}} \cdot D_{\text {round }} \cdot D_{\mathrm{arg}} \cdot D_{\mathrm{fpfn}}(\mu)\) and \(C=\tilde{\phi}_{\mathrm{fp}}(\mu) \sqrt{D_{\mathrm{arg}} \cdot D}\) with
  \begin{align*}
    D_{\mathrm{eq}} & =\exp\left({4\,{\left(\frac{\pi \sigma_{s} \ell}{q}\right)}^{2}}\right),
    \qquad D_{\mathrm{round}} =\prod_{\bar{s} \neq 0}{\left(\frac{\sin \left(\frac{\pi \bar{s}}{p}\right)}{\frac{\pi \bar{s}}{p}}\right)}^{-2 k_{\mathrm{fft}} \chi_{s}(\bar{s})} \\
    D_{\mathrm{arg}} &=\frac{1}{2}
                       \,\exp\left(2\,
                       {\left(\chi_{e}(0)
                       + e^{-\frac{8 \pi^{2} \alpha^{-2} \ell^{2}}{q^{2\,(m + k_\rest)}}}\right)}^{m}
                       \cdot {\left(\chi_{s}(0)+e^{-\frac{8 \pi^{2} \ell^{2}}{q^{2(m+k_\rest)}}}\right)}^{k_\rest}\right) \approx \frac{1}{2} \\
    D_{\mathrm{fpfn}}(\mu) &={\left({\phi}_{\mathrm{fp}}(\mu)+{\phi}_{\mathrm{fn}}(\mu)\right)}^{2} \cdot \mu \\
    {\phi}_{\mathrm{fp}}(\mu) &=\Phi^{-1}\left(1-\frac{\mu}{2 \cdot {G\left(\chi_{s}\right)} \cdot p^{k_{\mathrm{fft}}}}\right) \quad \text{ and } \quad  \tilde{\phi}_{\mathrm{fn}}(\mu) =\Phi^{-1}\left(1-\frac{\mu}{2}\right).
  \end{align*}
  Then, with probability at least $1-\nu$ the algorithm returns $\vec{s}_\enum$.
\end{lemma}

\noindent The proof is exactly the same as in~\cite[Theorem~5.9]{Matzov22} except that we replace the inequality \[ \Exp{}{N_\enum(\vec{s}_{\enum})}\leqslant 2^{k_\enum H(\chi_s)} \] by \[ \Exp{}{N_\enum(\vec{s}_{\enum})}=G(\chi_s^{k_\enum}). \] The reason for this replacement is that the first inequality does not appear to be justified in~\cite{Matzov22} and does not hold in general. See \Cref{sec:guessing} for more details. We can now state our main theorem which gives the complexity of our quantum augmented dual attack.

\begin{theorem}\label{thm:quantum_dual_complexity}
  Let $(n,m,q,\chi_s,\chi_s)$ be $\LWE$ parameters, $(\beta_0,\beta_1,k_\enum,k_\fft,k_\rest,p)$ be a partial tuple of parameters for \Cref{alg:quantum_dual_attack}, and let $0<\nu<1$. Let \(T_{\mathrm{BKZ}}(d,\beta_0)\) denote the cost of running BKZ-\(\beta_{0}\) in dimension \(d\), let \(T_{\mathrm{sieve}}(\beta_1)\) be the cost of sieving in dimension \(\beta_{1}\) and let \(G^{qc}(X)\) be the quantum guessing complexity. Choosing the parameters $C,D$ according to \Cref{lem:correctness_quantum_all}, \Cref{alg:quantum_dual_attack} outputs $\vec{s}_{\enum}$ with probability at least $1-\nu$ in expected time
  \[    O\left(\left\lceil\frac{D}{(\sqrt{4/3})^{\beta_1+o(\beta_1)}}\right\rceil\cdot\left(T_{\mathrm{BKZ}}(d,\beta_0)+T_{\mathrm{sieve}}(\beta_1)\right)+G^{qc}(\chi_s^{k_\enum})\cdot p^{k_\fft/2} \cdot \sqrt{D}\right)
\]
    up to polynomial factors.
\end{theorem}

\begin{proof}
The correctness follows directly from \Cref{lem:correctness_quantum_all}. We now analyse the
complexity. In this proof, we will neglect all polynomials factors in $n$ and $m$.

First, the cost of Step~\ref{alg:line:bkz} is that of \Cref{alg:short_vectors_sampling}, which is
\begin{equation}\label{eq:th7:cost_create_list} \left\lceil\frac{D}{N_{\mathrm{sieve}}(\beta_1)}\right
	\rceil\cdot\left(T_{\mathrm{BKZ}}(d,\beta_0) +T_{\mathrm{sieve}}(\beta_1)\right)
\end{equation} and $N_{\mathrm{sieve}}(\beta_1)=(\sqrt{4/3})^{\beta_1+o(\beta_1)}$ (this was justified in
\cite{AC:GuoJoh21}). Next, the cost of creating (Step~\ref{alg:quantum_dual_attack:create_qram}) and
filling (Steps~\ref{alg:quantum_dual_attack:loop_qram}-
the QRACM is $D$ (up to polynomial factors), by the assumption of the QRACM model (see
\Cref{quantum_memory}), which is negligible compared to \eqref{eq:th7:cost_create_list}. The cost of
Step~\ref{alg:quantum_dual_attack:create_quant_est_alg} is zero (all the cost of the algorithm is
when $\mathcal{A}$ is run), this step is just for the readability of the algorithm.

We now analyse the complexity of each call to the oracle $\mathcal{O}$. On input $\tilde{s}_\enum$, the
algorithm starts by creating the oracle $\hat{\mathcal{O}}$. This creation cost is negligible, all the cost is
incurred when the oracle is run. We need to analyse the cost of each call to $\hat{\mathcal{O}}$. The cost of
computing $\theta$ (Step~\ref{alg:quantum_dual_attack:compute_theta}) is negligible (polynomial). The
creation of oracle $\mathcal{O}'_W$ is also negligible and each call to $\mathcal{O}'_W$ takes time polynomial since in
our model, getting an entry from the QRACM $\mathcal{O}_W$ is $O(1)$. Therefore, the cost of each call to
$\hat{\mathcal{O}}$ is, up to polynomial factors, the cost of running $\mathcal{A}^{\mathcal{O}'_W}$. By
\Cref{thm:Quantum_estimation_f_W}, this algorithm makes $O(\varepsilon^{-1}\log\tfrac{1}{\delta})$ queries to
$\mathcal{O}'_W$ (which runs in polynomial time) and takes time $O(\varepsilon^{-1}\log(\tfrac{1}{\delta})\poly[\log n])$.
Recall that by Step~\ref{alg:quantum_dual_attack:create_quant_est_alg}, we have $\delta=\tfrac{1}{10}$
and $\varepsilon=\tfrac{C}{D}n$. Furthermore, recall that we chose the parameter $C$ according to
\Cref{lem:matzov-c-d}, that is
\[ C = \phi_{\mathrm{fp}}(\mu) \cdot \sqrt{D_{\mathrm{arg}} \cdot D}, \qquad D_{\arg} =
\frac{1}{2}+\exp\left({-8{\left(\frac{\pi \tau}{q}\right)}^{2}}\right)
\]
It is immediate that $D_{\arg}\geqslant 1/2$ by definition, and $\phi_{\mathrm{fp}}(\mu)$ is a constant
factor (it only depends $\mu=\tfrac{\nu}{2}$ which is fixed). Therefore, each call to $\hat{\mathcal{O}}$ takes
time
\begin{equation}\label{eq:th7:cost_Ohat} O\left(\frac{D}{C}\poly[\log n]\right)=
	O\left(\sqrt{D}\poly[\log n]\right).
\end{equation}
We can now finish the analysis of the cost of running $\mathcal{O}$. The cost of the search
for an FFT coefficient above the threshold (Step~\ref{alg:quantum_dual_attack:grover_for_fft}) is
given by \Cref{thm:hoyer03}: it makes $O(\sqrt{M})$ calls to $\hat{\mathcal{O}}$ where $M$ is the size of the
search space, which is $\ZZ_p^{k_\fft}$. Hence, the total cost of each call to $\mathcal{O}$ is
\begin{equation}\label{eq:th7:cost_O} O(\sqrt{p^{k_\fft}})
	=O\left(p^{k_\fft/2}\sqrt{D}\poly[\log n]\right)
\end{equation}

It remains to analyse the final cost: creating the oracle $\tilde{\mathcal{O}}$ takes polynomial time. The
cost of each call to the oracle is \eqref{eq:th7:cost_O} plus a polynomial cost. Indeed, finding the
$i^{th}$ most probably vector according to $\chi_s^{k_\enum}$ can easily be done in polynomial time for
all reasonable choices of $\chi_s$ such as the discrete Gaussian. Finally, the final search
(Step~\ref{alg:quantum_dual_attack:guess}), given by \Cref{lem:quantum_guessing}, takes an expected
time $G^{qc}(\chi_s^{k_\enum})$ and makes an expected $G^{qc}(\chi_s^{k_\enum})$ calls to $\mathcal{O}$. Therefore
the overall cost of this step is
\begin{equation}\label{eq:th7:search} O(G^{qc}(\chi_s^{k_\enum})\cdot\eqref{eq:th7:cost_O})
	=O\left(\sqrt{D}p^{k_\fft/2}G^{qc}(\chi_s^{k_\enum})\poly[\log n]\right).
\end{equation}

To conclude the proof, just note that the total cost of the algorith, up to
polynomial factors is $\eqref{eq:th7:cost_create_list}+\eqref{eq:th7:search}$.\qed{}

\end{proof}

\section{Application}\label{sec:application}

In this section we give \emph{rough} estimates for the impact of our algorithm on the cost of solving lattice parameters from the literature.
In particular, we consider NIST PQC Round 3 candidate Saber~\cite{NISTPQC-R3:SABER20} and NIST PQC to-be-standardised candidate Kyber~\cite{NISTPQC-R3:CRYSTALS-KYBER20} and some TFHE
parameters~\cite{JC:CGGI20}.\footnote{These schemes do not use discrete Gaussians as secret distributions. However, we simply ignore this difference here.} Such a measurement is complicated by two major obstacles.
\begin{itemize}
\item The cost given in \Cref{thm:quantum_dual_complexity} is the sum of two costs: lattice reduction and a quantum search. Roughly speaking, lowering the first summand increases the second and vice versa. In other words, the final cost is obtained by balancing the two summands. For the first summand cost estimates in various cost models are available. In particular, estimates in quantum circuit models are available~\cite{AC:AGPS20}. Thus, to give precise cost estimates we require quantum circuit costs for the oracles in \Cref{alg:quantum_dual_attack}. It is clear that such costs would be substantial when compared with e.g.~\cite{AC:AGPS20}. In the latter, the costed circuit is essentially an XOR followed by an adder. Here, we have to implement matrix vector products mod \(q\) which will cost significantly more. We consider designing and costing quantum circuits for elementary operations such as linear algebra mod \(q\) beyond the scope of this paper. For this reason, we cost our algorithm in the quantum query model only, both for the oracle inside lattice reduction and our quantum search. In this model, all oracle queries are assigned unit cost. As just outlined, this is unrealistic but gives a ``best case'' estimate from the perspective of an attacker. Similarly, in the same spirit as the ``Core-SVP'' model, we cost algorithms by simply dropping the \(O()\) around cost estimates, which means our estimates may over- or underestimate the true cost; however, all evidence so far points towards underestimating the costs.
\item A second major obstacle is that our algorithm critically relies on QRACM, a possibly unrealistic resource as already pointed out in e.g.~\cite{AC:AGPS20}. Thus, even armed with a  quantum circuit for our oracle, we would have to assume a QRACM oracle (for which we, following previous work~\cite{AC:AGPS20}, assign unit cost for querying). This would not permit us to draw conclusions about realistic costs of solving instances of lattice problems.
\end{itemize}

\begin{table}[htbp]
\centering
\caption{Dual attack cost estimates. All costs are logarithms to base two.}\label{tab:results}
\begin{tabular}{l@{\hskip 1em}r@{\hskip 1em}r@{\hskip 1em}r@{\hskip 1em}r@{\hskip 1em}r@{\hskip 1em}r@{\hskip 1em}r@{\hskip 1em}r}
\toprule
 Scheme     & CC    & CN    & C0    & GE19  & QN    & Q0    & This work & This work \\
            &       &       &       &       &       &       & (QN)      & (Q0) \\
  \midrule
 Kyber 512  & 139.2 & 134.4 & 115.4 & 139.5 & 124.4 & 102.7 & 119.3 & 99.7 \\
 Kyber 768  & 196.1 & 190.6 & 173.7 & 191.9 & 175.3 & 154.6 & 168.3 & 150.0 \\
 Kyber 1024 & 262.4 & 256.1 & 241.8 & 252.0 & 234.5 & 215.0 & 225.6 & 208.4 \\
 LightSaber & 145.3 & 140.8 & 121.8 & 145.3 & 129.7 & 107.7 & 120.8 & 100.6 \\
 Saber      & 204.7 & 198.9 & 182.3 & 199.0 & 182.6 & 162.4 & 176.6 & 157.9 \\
 FireSaber  & 267.9 & 261.7 & 247.4 & 257.0 & 239.4 & 220.3 & 231.3 & 214.0 \\
 TFHE630    & 118.2 & 113.3 & 93.0  & 120.2 & 105.2 & 83.0  & 100.8 & 80.7 \\
 TFHE1024   & 122.0 & 117.2 & 95.4  & 123.9 & 108.5 & 84.8  & 105.6 & 83.2 \\
 \bottomrule
 \end{tabular}
 \end{table}

 \noindent We give the source code and the results of the comparison in \cref{sec:source-code} and \cref{tab:results}.\footref{fn:bug} In our table, for each set of parameters, we give the following cost estimates.
\begin{description}
\item[CC] Classical cost estimates in a classical circuit model~\cite{AC:AGPS20,NISTPQC-R3:CRYSTALS-KYBER20,Matzov22} for \Cref{alg:matzov} using~\cite{SODA:BDGL16} as the sieving oracle. We derive these estimates by implementing the cost estimates from~\cite{Matzov22}, those tagged ``asymptotic'' cf.~\cite{MABMDAP22}.\footnote{This explains the minor differences in numerical results compared to~\cite{Matzov22}. In particular, we have an additional exponential factor for the guessing complexity, cf.~\Cref{lem:guessing_complexity_of_discrete_gaussian_multidim}.} This is the most detailed cost estimate available in the literature. However, we caution that these estimates, too, ignore the cost of memory access and thus may significantly underestimate the true cost. That is, while RAM access is expected to be considerably cheaper than QRACM it is still not ``free'', cf.~\cite{MABMDAP22}. This cost model is called ``list\_decoding-classical'' in~\cite{AC:AGPS20}. We naturally do not cost our algorithm in this cost model.
\item[CN] Classical cost estimates in a query model for \Cref{alg:matzov} using~\cite{SODA:BDGL16} as the sieving oracle. We include this cost model for completeness and for interpreting our quantum query cost model estimates. This cost model is called ``list\_decoding-naive\_classical'' in~\cite{AC:AGPS20}. We naturally do not cost our algorithm in this cost model.
\item[C0] Classical cost estimates in the ``Core-SVP'' cost model~\cite{USENIX:ADPS16} for \Cref{alg:matzov} using~\cite{SODA:BDGL16} as the sieving oracle. This model assumes a single SVP call suffices to reduce a lattice. It furthermore assumes that all lower-order terms in the exponent are zero. This is to enable comparison with ``Q0'' below.
\item[GE19] Quantum costs in a circuit model based on~\cite{GidEke19} for \Cref{alg:matzov} using~\cite{SODA:BDGL16}. This is the most detailed quantum cost model available in the literature but we recall that here we still assume unit cost QRACM\@. This cost model is called ``list\_decoding-ge19'' in~\cite{AC:AGPS20}. We do not cost our algorithm in this cost model due to the lack of a quantum circuit design for our oracles.\footnote{Note that the quantum costs in this cost model may be higher than classical costs because the crossover between classical and quantum computing under these cost models is in higher dimensions.}
\item[QN] Quantum costs in the quantum query model for \Cref{alg:matzov} using the quantum version of~\cite{SODA:BDGL16} as the sieving oracle; all other steps are classical. This cost model is called ``list\_decoding-naive\_quantum'' in~\cite{AC:AGPS20}.
\item[Q0] Quantum cost estimates in the ``Core-SVP'' cost model~\cite{USENIX:ADPS16} for \Cref{alg:matzov} using~\cite{AC:ChaLoy21} as the sieving oracle; all other steps are classical. This is the asymptotically fastest quantum sieving algorithm but no estimates exist in the literature for lower-order terms; hence, we only consider it in the Core-SVP model.
\item[This work (QN)] The cost of \Cref{alg:quantum_dual_attack} in the quantum query model assuming the quantum version~\cite{PhD:Laarhoven15,AC:AGPS20} of~\cite{SODA:BDGL16}. Thus, the most natural comparison is to the column labelled ``CN''.
\item[This work (Q0)] The cost of \Cref{alg:quantum_dual_attack} in the Core-SVP model assuming~\cite{AC:ChaLoy21}. Thus, the most natural comparison is to the column labelled ``C0''.
\end{description}

\begin{remark}
  It is worth noting that our optimisation routine almost always returns attack parameters where \(k_{\fft} = 0\) and \(k_{\enum} > 0\), i.e.~the FFT part of the algorithm is disabled, leaving the mean estimation and secret guessing part.
\end{remark}

On the one hand, comparing the column labelled ``QN'' and the last column shows that our algorithm offers an improvement of between 3 and 9 ``bits'' in complexity in the query model. On the other hand, even in this -- arguably unrealistic -- model our improvements do not lower the cost of solving below a square-root of the targeted security level. That is, to force a revision of lattice parameters, a quantum algorithm would have to obtain a quadratic speed-up over the classical cost given as ``CN''.

\section{Open Problems}\label{sec:open_problem}

The crux of our quantum improvement is \Cref{sec:quantum_dual_attack}. Here we formalise the problem
that this algorithm solves a promise variant. We introduce some minor notation first.
Fix a finite group \(G=\ZZ_q^n\) and a list
\[
  L=\{(\vec{u}_0,w_0),\ldots,(\vec{u}_{k-1},w_{k-1})\}
\]
where \(\vec{u}_j\) are elements of $G$ and $w_j$ complex numbers.
For every $\vec{x}\in G$,
we let
\[
    f_L(\vec{x})=\sum_{j=0}^{k-1} w_je^{-\tfrac{2i\pi}{q}\vec{u}_j^T\vec{x}}.
\]
In other words, $f_L$ is the (unormalized) discrete Fourier transform of
the sequence $T:G\to\CC$ defined by
\begin{equation}\label{eq:def_T_for_Fourier}
    T(\vec{u})=\sum_{j:\vec{u}_j=\vec{u}}w_j.
\end{equation}
We now introduce two problems,
which we call ``input sparse FFT'' to avoid confusion with
``sparse FFT'' where the
sparseness refers to the number of nonzero Fourier coefficient,
not the number of nonzero inputs coefficients.
\medskip

\newcommand{\SparseFFTThreshold}{\texttt{INPUT}-\texttt{SPARSE}-\texttt{FFT}-\texttt{THRESHOLD}}
\begin{samepage}
\noindent\SparseFFTThreshold{}:
\begin{itemize}
    \item \textbf{input:} $G=\ZZ_q^n$ a finite group,
    \item \textbf{input:} $\delta>0$ a threshold,
    \item \textbf{input:} $L=\{(\vec{u}_0,w_0),\ldots,(\vec{u}_{k-1},w_{k-1})\}$,
    \item \textbf{output:} decide whether $\exists\vec{x}\in G$ such that $\Re({f_L(\vec{x})})>\delta$.
\end{itemize}
\end{samepage}

\newcommand{\PromiseSparseFFTThreshold}{\texttt{PROMISE}-\texttt{INPUT}-\texttt{SPARSE}-\texttt{FFT}-\texttt{THRESHOLD}}
\begin{samepage}
\noindent\PromiseSparseFFTThreshold{}:
\begin{itemize}
    \item \textbf{input:} $G=\ZZ_q^n$ a finite group,
    \item \textbf{input:} $\delta^+>\delta^->0$ two thresholds,
    \item \textbf{input:} $L=\{(\vec{u}_0,w_0),\ldots,(\vec{u}_{k-1},w_{k-1})\}$,
    \item \textbf{promise:} $\Re({{f_L}(\vec{x})})\notin[\delta^-,\delta^+]$ for all $\vec{x}\in G$,
    \item \textbf{output:} decide whether $\exists\vec{u}\in G$ such that $\Re({{f_L}(\vec{x})})>\delta^{+}$.
\end{itemize}
\end{samepage}

\begin{remark}
  To map this formulation back to our task,
  let \(\vec{u}_{j}=\lfloor \frac{p}{q} \cdot \vec{y}_{j,\fft}\rceil\) and
  \(w_j=\exp\left((\vec{x}_j^T \cdot \vec{b} -\vec{y}^T_{j,\enum}\cdot \tilde{\vec{s}}_\enum)\cdot \tfrac{2i\pi}{q}\right)\).
  \cref{alg:matzov:line:open-problem} of \cref{alg:matzov} is then
  adding $w_j$ to the cell $\vec{u}_j$ of \(T\), hence building $T$
  as defined in \eqref{eq:def_T_for_Fourier}.
  We then seek to decide if there is some \(\vec{x} = \tilde{\vec{s}}_{\fft} \in G = \ZZ_{p}^{k_\fft}\) s.t. \(\Re({f_L(\vec{x})})>\delta = C\), i.e.~the entry in the FFT'd table \(T\).
\end{remark}

\subsection*{Classical case}

In the classical case, to the best of our knowledge, the best algorithm is to
first fill an multidimensional array $T$ with the $k$ coefficients
and perform a complete FFT on the $|G|$ coefficients, which therefore takes time
$O(|G|\log|G|)$ plus $O(k)$ to fill the array.
While there are algorithms for ``sparse'' FFT (see e.g.~\cite{10.1145/2213977.2214029}), it is not clear that their approximation guarantees
would be sufficient. Indeed, the sparseness in such algorithm refers to the number of
output coefficients $\widehat{f_L}(\vec{u})$ which is assumed small. Since we expect all the output
coefficients of our FFT to be small and the threshold $\delta$ to be small compared to the number of input coefficients,
it is unlikely that such an approximation would be sufficient.

In the quantum case, the situation depends on the availability of quantum memories (QRACM).
Our algorithm relies on the use of a QRACM in a crucial way. In fact, without a QRACM,
we are not aware of any algorithm better than the classical one for all ranges of parameters. This is surprising in light
of the fact that QFT can be done in polynomial time: we will explain below why this fact alone is not
sufficient.

\subsection*{With access to a QRACM}

Our quantum (with QRACM) algorithm from \Cref{sec:quantum_dual_attack} solves
\PromiseSparseFFTThreshold{} as follows. For every $\vec{x}\in G$,
it computes an approximation of $\Re({{f_L}(\vec{x})})$ with error at most
$\tfrac{1}{2}(\delta^+-\delta^-)$ and then compares it to $\delta^+$. By the promise,
this suffices to solve the problem. We then leverage two facts to obtain a quantum speedup:
the search over $\vec{x}\in G$ can be done using Grover's algorithm, and the approximation
is done by amplitude estimation (\Cref{thm:Quantum_estimation_f_W}).
The running time of our algorithm is $k\sqrt{|G|}/{(\delta^+-\delta^-)}$
assuming oracle access to the $\vec{u}_j$ and $w_j$,
and it outputs a correct index with constant
probability. A QRACM holding the $\vec{u}_j$ and $w_j$ can be built in time $O(k)$.
In the dual attack algorithm, the interesting set of parameters for this
algorithm is roughly
\begin{equation}\label{eq:setting_dual_attack}
    \delta^+-\delta^-=\Theta(\sqrt{k}),
    \qquad
    \delta^-=\Theta(\sqrt{k}),
    \qquad
    |w_j|=1.
\end{equation}
Therefore our algorithm has running time
roughly $O(\sqrt{k|G|})$ (assuming oracle access to the $\vec{u}_j$
and $w_j$) which is always better than
$O(|G|\log|G|+k)$ and potentially
much better if $k$ is either much smaller or much bigger
than $|G|$.

\subsection*{Without access to a QRACM}

It is possible to obtain a speed up over the classical
algorithm when $k$ is much smaller than $|G|$ without using
any QRACM. First, one can create an oracle $\mathcal{O}_L$
that given $\vec{x}$, compute ${f_L}(\vec{x})$
and returns $1$ if it is greater than $\delta$.
This oracle contains the list $L$ hardcoded into the circuit
and compute the sum naively.
Hence, it has depth $O(k)$ and takes time $O(k)$ to evaluate.
We then use Grover's algorithm to search for a $\vec{u}$
such that $\mathcal{O}(\vec{x})=1$, that is
${f_L}(\vec{u})>\delta$. Grover's search
and the oracle can be implemented with a plain quantum
circuit and a polynomial number of qubits, without the need
for a QRACM. The overall complexity is then, up to some small
polynomial factors for the arithmetic operations that we have
neglected,
\[
    O(k\sqrt{|G|}).
\]
This algorithm is better than the classical $O(|G|\log|G|)$
whenever $k=o(\sqrt{|G|})$.
In the context of dual attacks, this unfortunately brings
no improvement for our algorithm
(without access to a QRACM) since the optimal choice of the
parameters always seem to satisfy $k\gtrsim\sqrt{|G|}$.
For this reason, we did not report the results in
\Cref{tab:results}.

\subsection*{Why the QFT does not help}

Let $L=\{(\vec{u}_0,w_0),\ldots,(\vec{u}_{k-1},w_{k-1})\}$ be a list.
In order to apply the QFT, we need
to create the superposition
\begin{equation}\label{eq:superposition_for_fft}
    \ket{\psi}=\frac{1}{\sqrt{Z}}\sum_{j=0}^{k-1}w_j\ket{\vec{u}_j}
\end{equation}
where $Z$ is the normalisation factor. In general,
\[
    Z=\sum_{\vec{u}\in G}\Bigg|\sum_{j:\vec{u}_j=\vec{u}}w_j\Bigg|^2
\]
which simplifies to
\begin{equation}\label{eq:value_Z_all_distinct}
    Z=\sum_{j=0}^{k-1}|w_j|^2
    \qquad\text{if the }\vec{u}_j\text{ are distinct}.
\end{equation}
We then perform a QFT on $\ket{\psi}$ to obtain
\[
    \ket{\widehat{\psi}}=\frac{1}{\sqrt{Z|G|}}\sum_{\vec{u}\in G}{f_L}(\vec{u})\ket{\vec{u}}.
\]
We now would like to decide if there is some $\vec{x}\in G$ such that $f_L(\vec{x})>\delta$. Unfortunately,
there are two problems with this approach:
\begin{enumerate}
\item it is not clear how to create the superposition in \cref{eq:superposition_for_fft} efficiently without a QRACM,
\item it is not clear that we can efficiently detect whether there is an amplitude in front of some $\ket{\vec{x}}$ which is above the threshold.
\end{enumerate}
The first point is not a problem in our setting since we consider QRACM as part of our model of computation. Indeed, by storing the $w_j$ in a QRACM and applying
some elementary operations, one can produce such a state with high probability
(see e.g.~\cite[Theorem~1]{PhysRevA.99.012301}).
Regarding the second point, a possible strategy would be to measure $\ket{\widehat{\psi}}$ and obtain one $\vec{x}$.
By repeating this algorithm a very large number of times,
we can approximate the probability of the most likely $\vec{x}$
and therefore detect whether there is some sufficiently large ${f_L}(\vec{x})$.
In order to approximate this quantity within $\varepsilon$
and with probability at least $1-\nu$,
we need $N$ samples, where (see \Cref{sec:open_question_calc})
\begin{equation}\label{eq:compl_fft_threshold_with_qram}
    N=O(\ln(1/\nu)/\varepsilon^2),
    \qquad
    \varepsilon=\frac{\eta^+-\eta^-}{2},
    \qquad
    \eta^{\pm}=\frac{(\delta^{\pm})^2}{Z\cdot |G|}.
\end{equation}
In the case of the dual attack,
the parameters \eqref{eq:setting_dual_attack}
do not completely characterize the state and we need to take
into account the relation between $k$ and $|G|$.
\textbf{Assuming that $k\ll|G|$,}
then the $\vec{u}_j$ are distinct with high probability and therefore,
by \eqref{eq:value_Z_all_distinct},
$Z=k$ since $|w_j|=1$ because all inputs coefficients.
Therefore,
\[
    \varepsilon
        =\frac{\eta^+-\eta^-}{2}
        =\frac{(\delta^+)^2-(\delta^-)^2}{2|G|Z}
        =\frac{(\delta^+-\delta^-)(\delta^++\delta^-)}{2|G|Z}
        =\Theta\left(\frac{1}{|G|}\right).
\]
Hence, to succeed with probability $1-\nu$, we need
\begin{equation}\label{eq:samples_fft_threshold_naive}
    N=\Theta(|G|^2\ln(1/\nu))
\end{equation}
samples. It is not clear if there is a better strategy that merely decides on the presence of some \(\vec{x}\) without recovering it. Note that \eqref{eq:samples_fft_threshold_naive}
is even worse than the classical complexity.
\medskip

\subsection*{Summary}

We have introduced a problem, and its promise version,
related to FFT on a sparse input.
We have explained how our quantum
algorithm for the dual attacks solves the promise version
more efficiently than the classical algorithm but requires
access to a QRACM. We have also given an alternative quantum
algorithm that does not require access to a QRACM but which
is always worse than our algorithm and not always better
than the classical algorithm (depending on the value of $k$).
Specifically, we have the following algorithms to solve the promise problem:
\begin{itemize}
    \item classical: $O(|G|\log|G|)$ plus $O(k)$ for array filling,
    \item quantum (no QRACM): $O(k\sqrt{|G|})$,
    \item quantum (QRACM): $O(k\sqrt{|G|/(\delta^+-\delta^-)})$ plus $O(k)$ for QRACM filling.
\end{itemize}
Which one is better depends on the choice of the parameters. Clearly, $\Omega(k)$ is a lower bound on the complexity since
the algorithm needs to read the input in any case.
We have also explained why the natural algorithm based on
the QFT and measurements is seemingly always worse than
all of the above.

In the context of dual attacks, the quantum algorithm with QRACM above
has complexity $O(\sqrt{k|G|})$ for parameters \eqref{eq:setting_dual_attack}.
Therefore the following two questions regarding \PromiseSparseFFTThreshold{}:
\begin{itemize}
\item Is the quantum complexity \(O(\sqrt{k|G|})\) optimal  with a QRACM
?
\item  Can we achieve quantum complexity better than \(O(|G|\log|G|+k)\) when $k=\Omega(\sqrt{|G|})$ and without a QRACM?
\end{itemize}

Another open problem is to prove a matching lower bound for the
(classical and quantum) guessing complexity of a discrete Gaussian.
See \Cref{lem:guessing_complexity_of_discrete_gaussian_multidim}
and the remark below for more details.

\ifeprint
\section*{Acknowledgements}
We thank André Chailloux, Yanlin Chen, Thomas Debris-Alazard, Yassine Ham\-oudi, André Schrottenloher, Amaury Pouly, Jean-Pierre Tillich for helpful discussions.
We thank Erik Mårtensson and Alessandro Budroni for alerting us to a bug in the estimation code contained and used in a previous version of this work.
The research of MA was supported by EPSRC grants EP/S020330/1 and EP/S02087X/1, and by the European Union Horizon 2020 Research and Innovation Program Grant 780701. The research of YS was supported by EPSRC grant EP/S02087X/1 and EP/W02778X/1.
\fi

\bibliographystyle{alpha}
\bibliography{cryptobib/abbrev3,cryptobib/crypto_crossref,local}

\newpage
\appendix

\ifeprint
\section*{Appendix}
\else
\section*{Supplementary Material}
\fi

\section{Proof of \Cref{lem:quantum_guessing}}\label{sec:proof_guessing_complexity}
We intuitively want to call the algorithm from \Cref{lem:quantum_find_first}
directly but there is a small issue. Indeed, with small probability, the algorithm
might not find the marked element and run in time $\sqrt{|E|}$ which will make average
execution time bigger than what we want. Instead, we search in a subset of the array
that we double at each failure. This ensures that the ``worst-case'' cost remains under
control.

\begin{algorithm}[ht]
    \SetKwProg{Oracle}{create oracle}{:}{}
    \KwIn{size $N$ of $E$, order $\sigma$ and bounded-error oracle $\mathcal{O}$}
    \KwOut{(guess for) $i$ such that $\mathcal{O}(i)=1$, or $\bot$}
    \Oracle{\textup{$\mathcal{O}'$($i$)}}{
        \Return{$\mathcal{O}(\sigma(i))$}
    }
    $n\gets 1$ \;
    \While{$n<2N$}{
        Use \Cref{lem:quantum_find_first} to find $i\in[1,\min(n,N)]$
        such that $\mathcal{O}'(i)=1$, or $i=\bot$ \;
        \If{$i\neq\bot$}{
            Call $1+2\log(n)$ times $\mathcal{O}'(i)$ and let $b\in\{0,1\}$ be the majority
            answer\;
            \If{$b=1$}{
                \Return{$i$}
            }
        }
        $n\gets 2n$\;
    }
    \Return{$\bot$}
    \caption{Quantum guessing algorithm with bounded error oracle\label{alg:quantum_guessing}}
\end{algorithm}

Consider \Cref{alg:quantum_guessing}, which we call with
$\sigma$ and $\mathcal{O}_x$ for some fixed $x\in E$.
Note that $\sigma$ is surjective so there
exists $i$ such that $\sigma(i)=x$ and then $\mathcal{O}'(i)=1$
with probability at least $9/10$. Since $\sigma$ is injective,
$\sigma(j)\neq x$ for $j\neq i$ so $\mathcal{O}'(j)=0$ with
probability at least $9/10$. Therefore, $\mathcal{O}'$ computes,
with bounded error at most $1/10$, the function $f_x(x)=1$ and $f_x(y)=0$
for all $y\neq x$.
Let $i_x$ be such that $\sigma(i_x)=x$, and $p_x$ such that $2^{p_x-1}<i_x\leqslant 2^{p_x}$.

Let $k<p_x$ and let us analyse the $k^{th}$ iteration of the loop.
During the iteration, $n=2^{k-1}$ so $n<i_x$ and therefore, $f_x(j)=0$ for all $j\in[1,n]$.
It follows that \Cref{lem:quantum_find_first} will return $\bot$ with probability at
least $9/10$ and then the loop will continue to its next iteration. With probability at
most $1/10$, it will return a (wrong) index $i\neq\bot$. But $f_x(i)=0$ so each call
to $\mathcal{O}'(i)$ will return $0$ with probability at least $9/10$. Therefore with
probability at most $10^{-k}$, a majority of the $1+2k$ calls will return $0$ and the
loop with continue. In summary, the loop will return a wrong index only with probability
at most $10^{-k-1}$. In all cases, the cost of this iteration is $O(\sqrt{2^k})$ queries/time
for the search, plus an extra $1+2k$ calls to the oracle.

It follows that with probability at most $\sum_{k=1}^{p_x-1}10^{-k-1}\leqslant 1/9$,
the algorithm will stop during on the first $p_x$ iterations and return a wrong answer.
Hence, with probability at least $8/9$, the algorithm reaches the $p_x^{th}$ iterations.

Now assume that the algorithm has reached the $p_x^{th}$ iteration.
Let $p_x\leqslant k\leqslant 1+\log_2(N)$ and let us analyse the $k^{th}$ iteration of the
loop.
During the iteration, $n=2^{k-1}$ so $i_x\in[1,\min(n,N)]$ and $f_x(j)=0$
for all $j\neq i_x$.
It follows that \Cref{lem:quantum_find_first} will return $i_x$ with probability at
least $9/10$. Since $f_x(i_x)=1$, each call to $\mathcal{O}'(i_x)$ will return $1$ with
probability at least $9/10$. Therefore with
probability at least $1-10^{-k}\geqslant 9/10$,
a majority of the $1+2k$ calls will return $1$
and the algorithm will return $i_x$.
In summary, with probability at least $(9/10)^2\geqslant 4/5$,
the algorithm stops and return $i_x$ during the $k^{th}$ iteration.
Conversely, with probability at most
$1/5$, the algorithm will either return a wrong answer or continue to the next
iteration.

This analysis shows that the algorithm is correct since already with probability at
least $8/9$ the algorithm will reach $p_x^{th}$ iteration and then with probability
at least $4/5$ it will return $i_x$ during that iteration, so it returns $i_x$
with probability at least $32/45$.

We now analyse the complexity: the cost of the $k^{th}$ iteration is the cost
of the search and (maybe) the $1+2k$ calls to the oracle, which overall is
\[
    O(\sqrt{2^{k-1}}+1+2k)=O(\sqrt{2^k})
\]
queries/time. To simplify the analysis, we can always assume that the first
$p_x-1$ iterations are done in all cases (this is clearly an upper bound in the case
where the algorithm stops too early). We can also assume that the while loop
does not stop at $n<2N$ but instead goes on forever (again this is clearly an upper
bound on the real cost).
For each $k\geqslant p_x$, the $k^{th}$ iteration
only occurs with probability at most $5^{-(k-p_x+1)}$ by the analysis above. Therefore,
the expected running time/number of queries of the algorithm is bounded by
\begin{align*}
    O\left(\sum_{k=1}^{p_x-1}\sqrt{2^k}+\sum_{k=p_x}^{\infty}5^{-(k-p_x+1)}\sqrt{2^k}\right)
        &=O\left(\sqrt{2^{p_x}}+\sqrt{2^{p_x}}\sum_{k=0}^\infty (\sqrt{2}/5)^k\right)\\
        &=O\left(\sqrt{2^{p_x}}\right)\\
        &=O\left(\sqrt{i_x}\right).
\end{align*}

Now recall that this analysis holds when the algorithm is called
with an oracle $\mathcal{O}_x$ for a given $x$.
We now let $x$ be drawn from $X$. Then every $x$ is chosen with
probability $\Prob{X}{X=x}$ and, when this is the case, the
returned index is $i_x=\sigma^{-1}(x)$.
Hence, the expected time/query
complexity of the algorithm when given $\mathcal{O}_X$ is
\[
    O\left(\sum_{x\in E}\Prob{X}{X=x}\sqrt{\sigma^{-1}(x)}\right)
    =O\left(\sum_{i=0}^\infty\Prob{X}{X=\sigma(i)}\sqrt{i}\right)
    =O(G^{qc}(X)).
\]
since we assumed that $\sigma$ orders the elements by decreasing
probability.
\qed{}

\section{Proof of Lemma~\ref{lem:guessing_complexity_of_discrete_gaussian_multidim}}\label{sec:proof:guessing_complexity_of_discrete_gaussian_multidim}

For reasons that will be clear in the proof of the bounds, we introduce
for all $s>0$ and $n\in\NN$,
\[
    g_{n,s}(\sigma):=\sum_{\ell=0}^\infty \frac{\rho_\sigma(\sqrt{\ell})}{\rho_\sigma(\ZZ^n)} \cdot N(\ell)^s
\]
where $N(\ell):=N_n(\ell)=\left|\set{\vec{x}\in\ZZ^n:\norm{\vec{x}}^2\leqslant\ell}\right|$. Note that $N(-1)=0$. We also bound $\rho_\sigma(\ZZ^n)$ by using the Poisson summation formula as follows:
\begin{equation}\label{eq:lower_bound_rho_Z^n}
    \rho_\sigma(\ZZ^n)
        ={\left(\rho_\sigma(\ZZ)\right)}^n
        ={\left(\sigma\sqrt{2\pi}\rho_{1/2\pi\sigma}(\ZZ)\right)}^n
        >{(\sigma\sqrt{2\pi})}^n
\end{equation}
and
\begin{equation}\label{eq:upper_bound_rho_Z^n}
    \rho_{1/2\pi\sigma}(\ZZ)
        =1+2\sum_{n=1}^\infty e^{-2\pi^2\sigma^2 n^2}
        \leqslant 1+2\sum_{n=1}^\infty e^{-2\pi^2\sigma^2 n}
        =\coth(\pi^2\sigma^2).
\end{equation}
Also recall~\cite{Ban93} that
\begin{equation}\label{eq:ineq_ban93_gaussian_weight}
    \rho_s(\mathcal{L})\leqslant s^n\cdot \rho(\mathcal{L})
\end{equation}
for any $s\geqslant 1$ and lattice $\mathcal{L}\subset\RR^n$.

\subsection{Generic bounds on $g_{n,k}(\sigma)$ for $k$ integer}

We start by obtaining a general relationship between the $g_{n,k}$ and $g_{n,1}$ where $k$ is an
integer. A simple counting argument shows that
${N(\ell)}^k\leqslant N_{kn}(k\ell)$. Indeed, the map
$\phi:(\vec{x}_1,\ldots,\vec{x}_k)\in(\ZZ^n)^k\to
\begin{pmatrix}\vec{x}_1,\cdots,\vec{x}_k\end{pmatrix}\in\ZZ^{kn}$ is injective
and if $\norm{\vec{x}_1}^2,\ldots,\norm{\vec{x}_k}^2\leqslant\ell$ then
$\norm{\phi(\vec{x}_1,\ldots,\vec{x}_k)}^2
=\norm{\vec{x}_1}^2+\cdots+\norm{\vec{x}_k}^2\leqslant k\ell$.
Therefore,
\[
    g_{n,k}(\sigma)
        \leqslant\sum_{\ell=0}^\infty
            N_{kn}(k\ell)\frac{e^{-\tfrac{\ell}{2\sigma^2}}}{\rho_\sigma(\ZZ^n)}
        =\frac{1}{\rho_\sigma(\ZZ^n)}H^0_{kn}(e^{-1/2\sigma^2})
\]
where
\[
    H^i_n(x)=\sum_{\ell=0}^\infty N_n(k\ell+i)x^{\ell}.
\]
Furthermore,
\[
    H_n^0(x^k)+xH_n^1(x^k)+\cdots+x^{k-1}H_n^{k-1}(x^k)=\sum_{\ell=0}^\infty N_n(\ell)\cdot x^\ell.
\]
Since $N_n(\ell)$ is increasing with $\ell$, we immediately have that
$H_n^0(x)\leqslant H_n^i(x)$ for all $i$ so
\[
    H_n^0(x)\leqslant \frac{1}{1+x^{1/k}+\cdots+x^{(k-1)/k}}
        \sum_{\ell=0}^\infty N_n(\ell)\cdot \sqrt[k]{x}^\ell.
\]
In particular,
\begin{align*}
    g_{n,k}(\sigma)
        &\leqslant \frac{1}{\rho_{\sigma}(\ZZ^{n})}\frac{1}{\sum_{i=0}^{k-1}e^{-i/2k\sigma^2}}
            \sum_{\ell=0}^\infty N_{kn}(\ell)\cdot e^{-\ell/2k\sigma^2}\\
        &= \frac{\rho_{\sqrt{k}\sigma}(\ZZ^{kn})}{\rho_{\sigma}(\ZZ^{n})}
            \frac{1}{\sum_{i=0}^{k-1}e^{-i/2k\sigma^2}}
            \sum_{\ell=0}^\infty \frac{N_{kn}(\ell)}{\rho_{\sqrt{k}\sigma}(\ZZ^{kn})}\cdot e^{-\ell/2(\sqrt{k}\sigma)^2}\\
        &=\frac{\rho_{\sqrt{k}\sigma}(\ZZ^{kn})}{\rho_{\sigma}(\ZZ^{n})}
            \frac{1}{\sum_{i=0}^{k-1}e^{-i/2k\sigma^2}}\cdot g_{kn,1}(\sqrt{k}\sigma).
\end{align*}
It is well-known (check by developing $(1-x)^{-1}$ into its series) that
\[
    \sum_{\ell=0}^\infty N_n(\ell)\cdot x^\ell=\frac{1}{1-x}
        \sum_{\ell=0}^\infty S_n^L(\ell)\cdot x^\ell,
    \qquad
    S_n^L(\ell)=\left|  \set{\vec{x}\in\ZZ^n:\norm{\vec{x}}^2=\ell}\right|.
\]
On the other hand, for $x=e^{-1/2\sigma^2}$, $\sum_{\ell=0}^\infty S_n^L(\ell)\cdot x^\ell=\rho_\sigma(\ZZ^n)$.
Therefore,
\begin{equation}\label{eq:bound_g1}
    g_{n,1}(\sigma)=\frac{1}{1-e^{-1/2\sigma^2}}
\end{equation}
Therefore
\begin{align*}
    g_{n,k}(\sigma)
        &\leqslant \frac{\sqrt{k}^{kn}\rho_{\sigma}(\ZZ)^{(k-1)n}}{\sum_{i=0}^{k-1}e^{-i/2k\sigma^2}}
            \cdot g_{kn,1}(\sqrt{k}\sigma)
        &&\text{by \eqref{eq:ineq_ban93_gaussian_weight}}\\
        &= \frac{\sqrt{k}^{kn}\rho_{\sigma}(\ZZ)^{(k-1)n}}{\sum_{i=0}^{k-1}e^{-i/2k\sigma^2}}
            \cdot \frac{1}{1-e^{-1/2(\sqrt{k}\sigma)^2}}
            &&\text{by \eqref{eq:bound_g1}}\\
        &=\frac{\sqrt{k}^{kn}\rho_{\sigma}(\ZZ)^{(k-1)n}}{1-e^{-1/2\sigma^2}}
            \numberthis\label{eq:bound_gn}.
\end{align*}
Note that this equation also holds for $k=1$ by \eqref{eq:bound_g1}. In fact, it even
holds for $k=0$, with the convention that $0^0=1$. since
\begin{equation}\label{eq:bound_g0}
    g_{n,0}(\sigma)
        =\frac{1}{\rho_\sigma(\ZZ^n)}\sum_{\ell=0}^\infty e^{-\ell/2\sigma^2}
        =\frac{1}{\rho_\sigma(\ZZ^n)}\frac{1}{1-e^{-1/2\sigma^2}}.
\end{equation}

\subsection{Bound on $G(D_{\ZZ^n,\sigma})$}

Observing that \(D_{\ZZ^n,\sigma}\) decreases with $\norm{\vec{x}}$, we rewrite the guessing complexity as
\[
    G(D_{\ZZ^n,\sigma})
        =\sum_{\ell=0}^\infty
            \frac{\rho_\sigma(\sqrt{\ell})}{\rho_\sigma(\ZZ^n)}\sum_{i=N(\ell-1)+1}^{N(\ell)}i.
\]
It follows that, for $n\geqslant 4$
(we need every number to be a sum of $n$ squares so that $N(\ell)>N(\ell-1)$):
\begin{align*}
    G(D_{\ZZ^n,\sigma})
        &=\sum_{\ell=0}^\infty
        \frac{\rho_\sigma(\sqrt{\ell})}{\rho_\sigma(\ZZ^n)} \cdot
            \frac{(N(\ell)-N(\ell-1)) (N(\ell-1)+N(\ell)+1)}{2}
            \label{eq:formula_G_X}\numberthis\\
        &\leqslant \sum_{\ell=0}^\infty
        \frac{\rho_\sigma(\sqrt{\ell})}{\rho_\sigma(\ZZ^n)} \cdot (N(\ell)-N(\ell-1))\cdot N(\ell)\\
        &\leqslant g_{n,2}(\sigma)\\
        &\leqslant \frac{2^{n}\rho_{\sigma}(\ZZ)^n}{1-e^{-1/2\sigma^2}}
            &&\hspace{-2cm}\text{by \eqref{eq:bound_gn}.}\label{eq:formula_G_X_upper}\numberthis
\end{align*}

\subsection{Bound on $G^{qc}(D_{\ZZ^n,\sigma})$}

For the quantum guessing complexity, we use the same approach to get that
\begin{equation}\label{eq:formula_Gqc_X}
    G^{qc}(D_{\ZZ^n,\sigma})
        =\sum_{\ell=0}^\infty
            \frac{\rho_\sigma(\sqrt{\ell})}{\rho_\sigma(\ZZ^n)}
            \sum_{i=N(\ell-1)+1}^{N(\ell)}\sqrt{i}.
\end{equation}
Now check by induction on $b$ that
\[
    \sum_{i=0}^b\sqrt{i}\leqslant \tfrac{2}{3}n^{3/2}+\tfrac{1}{2}\sqrt{n}.
\]
It follows that
\begin{align*}
      G^{qc}(D_{\ZZ^n,\sigma})
        &\leqslant \frac{2}{3}\sum_{\ell=0}^\infty
                \frac{\rho_\sigma(\sqrt{\ell})}{\rho_\sigma(\ZZ^n)}N(\ell)^{3/2}
                +\frac{1}{2}\sum_{\ell=0}^\infty
                \frac{\rho_\sigma(\sqrt{\ell})}{\rho_\sigma(\ZZ^n)}\sqrt{N(\ell)}\\
        &=\tfrac{2}{3}g_{n,3/2}(\sigma)+\tfrac{1}{2}g_{n,1/2}(\sigma).
\end{align*}
Now check, that for $\alpha=2/3$,
\begin{align*}
    g_{n,3/2}(\sigma)
    &=\sum_{\ell=0}^\infty
        \frac{\rho_\sigma(\sqrt{\ell})}{\rho_\sigma(\ZZ^n)}N(\ell)^{3/2}\\
    &=\frac{1}{\rho_\sigma(\ZZ^n)}\sum_{\ell=0}^\infty
        \sqrt{\rho_\sigma(\sqrt{\ell})^\alpha N(\ell)}^3\\
    &\leqslant\frac{1}{\rho_\sigma(\ZZ^n)}\sqrt{\sum_{\ell=0}^\infty
        \rho_\sigma(\sqrt{\ell})^\alpha N(\ell)}^3
        &&\text{by Cauchy–Schwarz}\\
    &=\frac{1}{\rho_\sigma(\ZZ^n)}\sqrt{\sum_{\ell=0}^\infty
        \rho_{\sigma/\sqrt{\alpha}}(\sqrt{\ell}) N(\ell)}^3\\
    &=\frac{1}{\rho_\sigma(\ZZ^n)}\sqrt{
        \rho_{\sigma/\sqrt{\alpha}}(\ZZ^n)g_{n,1}(\tfrac{\sigma}{\sqrt{\alpha}})}^3\\
    &=\frac{\rho_{\sigma/\sqrt{\alpha}}(\ZZ^n)^{3/2}}{\rho_\sigma(\ZZ^n)}
        \frac{1}{(1-e^{-1/3\sigma^2})^{3/2}}
        &&\text{by \eqref{eq:bound_g1}}\\
    &\leqslant\frac{\rho_{\sigma}(\ZZ^n)^{3/2}}{\sqrt{\alpha}^{3n/2}\rho_\sigma(\ZZ^n)}
        \frac{1}{(1-e^{-1/3\sigma^2})^{3/2}}
        &&\text{by \eqref{eq:ineq_ban93_gaussian_weight} since $1/\sqrt{\alpha}>1$}\\
    &=\left(\frac{3}{2}\right)^{3n/4}
        \frac{\sqrt{\rho_{\sigma}(\ZZ^n)}}{(1-e^{-1/3\sigma^2})^{3/2}}
        \numberthis\label{eq:upper_bound_gn3_2_unsimplified}
\end{align*}
Similarly,
\begin{align*}
    g_{n,1/2}(\sigma)
    &=\sum_{\ell=0}^\infty
        \frac{\rho_\sigma(\sqrt{\ell})}{\rho_\sigma(\ZZ^n)}N(\ell)^{1/2}\\
    &=\sum_{\ell=0}^\infty
        \sqrt{\frac{\rho_\sigma(\sqrt{\ell})}{\rho_\sigma(\ZZ^n)}N(\ell)}
        \cdot
        \sqrt{\frac{\rho_\sigma(\sqrt{\ell})}{\rho_\sigma(\ZZ^n)}}\\
    &\leqslant\sqrt{\sum_{\ell=0}^\infty\frac{\rho_\sigma(\sqrt{\ell})}{\rho_\sigma(\ZZ^n)}N(\ell)}
        \sqrt{\sum_{\ell=0}^\infty\frac{\rho_\sigma(\sqrt{\ell})}{\rho_\sigma(\ZZ^n)}}\\
    &=\sqrt{g_{n,1}(\sigma)g_{n,0}(\sigma)}\\
    &=\frac{1}{1-e^{-1/2\sigma^2}}\frac{1}{\sqrt{\rho_\sigma(\ZZ^n)}}
        &&\text{by \eqref{eq:bound_g1} and \eqref{eq:bound_g0}.}
\end{align*}
It follows by \eqref{eq:upper_bound_gn3_2_unsimplified} that
\begin{align*}
    G^{qc}(D_{\ZZ^n,\sigma})
        &\leqslant \frac{2}{3}\cdot\left(\frac{3}{2}\right)^{3n/4}
            \frac{\sqrt{\rho_{\sigma}(\ZZ^n)}}{(1-e^{-1/3\sigma^2})^{3/2}}
            +\frac{1}{2}\cdot\frac{1}{1-e^{-1/2\sigma^2}}\frac{1}{\sqrt{\rho_\sigma(\ZZ^n)}}\\
        &=\frac{2}{3}\cdot\left(\frac{3}{2}\right)^{3n/4}
            \frac{\sqrt{\rho_{\sigma}(\ZZ^n)}}{(1-e^{-1/3\sigma^2})^{3/2}}\\
        &\hspace{2ex}\times\left(1+\frac{3}{4}\cdot\left(\frac{2}{3}\right)^{3n/4}
            \frac{(1-e^{-1/3\sigma^2})^{3/2}}{1-e^{-1/2\sigma^2}}\frac{1}{\rho_\sigma(\ZZ^n)}\right)\\
        &\leqslant\frac{2}{3}\cdot\left(\frac{3}{2}\right)^{3n/4}
            \frac{\sqrt{\rho_{\sigma}(\ZZ^n)}}{(1-e^{-1/3\sigma^2})^{3/2}}
            \left(1+\frac{3}{4}\cdot\left(\frac{2}{3}\right)^{3n/4}
            \frac{1}{\rho_\sigma(\ZZ^n)}\right)\\
        &\leqslant\frac{2}{3}\cdot\left(\frac{3}{2}\right)^{3n/4}
            \frac{\sqrt{\rho_{\sigma}(\ZZ^n)}}{(1-e^{-1/3\sigma^2})^{3/2}}
            \left(1+\frac{3}{4}\cdot
            \left(\frac{2^{1/4}}{3^{3/4}\sigma\sqrt{\pi}}\right)^n\right)
            &&\text{by \eqref{eq:lower_bound_rho_Z^n}}\\
        &\leqslant\frac{7}{6}\cdot\left(\frac{3}{2}\right)^{3n/4}
            \frac{\sqrt{\rho_{\sigma}(\ZZ^n)}}{(1-e^{-1/3\sigma^2})^{3/2}}
            &&\hspace{-3cm}\text{when }\sigma\geqslant \sqrt[4]{\frac{2}{27\pi^2}}.
            \numberthis\label{eq:upper_bound_Gc_unsimplified}
\end{align*}

\subsection{Bound on $H(D_{\ZZ^n,\sigma})$}
Finally, we estimate the entropy of this distribution as follows:
\begin{align*}
    H(D_{\ZZ^n,\sigma})
        &=-\sum_{x\in\ZZ^n}D_{\ZZ^n,\sigma}(x)\log_2(D_{\ZZ^n,\sigma}(x))\\
        &=\frac{1}{2\sigma^2\log(2)\rho_\sigma(\ZZ)}
            \sum_{x\in\ZZ^n}\rho_\sigma(x)\cdot \|x\|^2
            +\frac{\log_2\rho_\sigma(\ZZ^n)}{\rho_\sigma(\ZZ^n)}
            \sum_{x\in\ZZ^n}\rho_\sigma(x)\\
        &=\frac{1}{2\sigma^2\log(2)\rho_\sigma(\ZZ^n)}
            \sum_{x\in\ZZ^n}\rho_\sigma(x)\cdot\|x\|^2
            +\log_2\rho_\sigma(\ZZ^n).
\end{align*}
Now note that
\begin{align*}
    \sum_{x\in\ZZ^n}\rho_\sigma(x)\cdot\|x\|^2
    &=\sum_{k=1}^n\sum_{x\in\ZZ^n}x_k^2\prod_{j=1}^k\rho_\sigma(x_j)\\
    &=\sum_{k=1}^n{\left(\sum_{x\in\ZZ}\rho_\sigma(x)\right)}^{n-1}
        \sum_{x\in\ZZ}x^2\rho_\sigma(x)\\
    &=n\,{\rho_\sigma(\ZZ)}^{n-1}\sum_{x\in\ZZ}f_\sigma(x)
\end{align*}
where $f_\sigma(x)=x^2\rho_\sigma(x)$. Hence, since $\rho_s(\ZZ^n)=\rho_s(\ZZ)^n$,
\begin{equation}\label{eq:formula_H_intermediate}
    H(D_{\ZZ^n,\sigma})=
        \frac{n}{2\sigma^2\log(2)\rho_\sigma(\ZZ)}
            \sum_{x\in\ZZ}f_\sigma(x)
            +n\log_2\rho_\sigma(\ZZ).
\end{equation}
It is not hard to check that the Fourier transform of $f_\sigma$ is
\[
    \widehat{f_\sigma}(x)=e^{-2\pi^2\sigma^2x^2}\sigma^3\sqrt{2\pi}\left(1-4\pi^2\sigma^2x^2\right).
\]
Therefore, by the Poisson summation formula
\begin{align*}
    \sum_{x\in\ZZ}f_\sigma(x)
        &=\sigma^3\,\sqrt{2\pi}\left(\sum_{x\in\ZZ}e^{-2\pi^2\sigma^2x^2}
            -4\pi^2\sigma^2\sum_{x\in\ZZ}e^{-2\pi^2\sigma^2x^2}x^2\right)\\
        &=\sigma^3\,\sqrt{2\pi}\left(\rho_{1/2\pi\sigma}(\ZZ)
            -4\pi^2\sigma^2\sum_{x\in\ZZ}e^{-2\pi^2\sigma^2x^2}x^2\right)\\
        &=\sigma^3\,\sqrt{2\pi}\left(\frac{1}{\sigma\sqrt{2\pi}}\rho_{\sigma}(\ZZ)
            -4\pi^2\sigma^2\sum_{x\in\ZZ}e^{-2\pi^2\sigma^2x^2}x^2\right)
            &&\text{by \eqref{eq:lower_bound_rho_Z^n}}\\
        &=\sigma^2\rho_{\sigma}(\ZZ)
            -4\sqrt{2\pi}\pi^2\sigma^5\sum_{x\in\ZZ}e^{-2\pi^2\sigma^2x^2}x^2
            \label{eq:sum_of_f_sigma}
            \numberthis
\end{align*}
Hence, we have that
\begin{align*}
    H(D_{\ZZ^n,\sigma})
        &\leqslant
            \frac{n}{2\log(2)}
            +n\log_2\rho_\sigma(\ZZ)
        &&\text{by \eqref{eq:formula_H_intermediate}
            and \eqref{eq:sum_of_f_sigma}}\\
        &\leqslant
            \frac{n}{2\log(2)}
            +n\log_2(\sigma\sqrt{2\pi}\coth(\pi^2\sigma^2))
        &&\text{by \eqref{eq:lower_bound_rho_Z^n} and
            \eqref{eq:upper_bound_rho_Z^n}.}
            \numberthis
\end{align*}
Now check that for any $\alpha>0$,
\[
    \sum_{x\in\ZZ}e^{-\alpha x^2}x^2
        \leqslant 2\sum_{n=1}^\infty e^{-\alpha n}n^2
        =\frac{2e^\alpha(1+e^\alpha)}{(e^\alpha-1)^3}
        \leqslant 4e^{-\alpha}.
\]
Hence, for $\alpha=2\pi^2\sigma^2$,
\[
    4\sqrt{2\pi}\pi^2\sigma^5\sum_{x\in\ZZ}e^{-2\pi^2\sigma^2x^2}x^2
        \leqslant
        16\sqrt{2\pi}\pi^2\sigma^5 e^{-2\pi^2\sigma^2}.
\]
From this, \eqref{eq:formula_H_intermediate} and \eqref{eq:sum_of_f_sigma}
we conclude that,
\begin{align*}
    H(D_{\ZZ^n,\sigma})
        &\geqslant
            \frac{n}{2\log(2)}
            +n\log_2\rho_\sigma(\ZZ)
            -\frac{n\cdot 16\sqrt{2\pi}\pi^2\sigma^5
            e^{-2\pi^2\sigma^2}}{2\sigma^2\log(2)\rho_\sigma(\ZZ)}\\
        &=
            \frac{n}{2\log(2)}
            +n\log_2\rho_\sigma(\ZZ)
            -n\frac{8\sqrt{2\pi}\pi^2\sigma^3
            e^{-2\pi^2\sigma^2}}{\log(2)\rho_\sigma(\ZZ)}\\
        &\geqslant
            \frac{n}{2\log(2)}
            +n\log_2\rho_\sigma(\ZZ)
            -n\frac{8\sqrt{2\pi}\pi^2\sigma^3
            e^{-2\pi^2\sigma^2}}{\log(2)\sigma\sqrt{2\pi}\coth(\pi^2\sigma^2)}
            &&\text{by \eqref{eq:upper_bound_rho_Z^n}}\\
        &=
            \frac{n}{2\log(2)}
            +n\log_2\rho_\sigma(\ZZ)
            -n\frac{8\pi^2\sigma^2
            e^{-2\pi^2\sigma^2}}{\log(2)\coth(\pi^2\sigma^2)}.
\end{align*}
Therefore,
\begin{equation}\label{eq:lower_bound_2pH}
    2^{H(D_{\ZZ^n,\sigma})}
        \geqslant e^{n/2}\rho_\sigma(\ZZ)^n e^{-n\frac{8\pi^2\sigma^2
            e^{-2\pi^2\sigma^2}}{\coth(\pi^2\sigma^2)}}.
\end{equation}

\subsection{Relation between $H(D_{\ZZ^n,\sigma})$, $G(D_{\ZZ^n,\sigma})$ and
$G^{qc}(D_{\ZZ^n,\sigma})$}

Recall that by \eqref{eq:formula_G_X_upper},
\[
    G(D_{\ZZ^n,\sigma})\leqslant
        \frac{2^{n}\rho_{\sigma}(\ZZ)^n}{1-e^{-1/2\sigma^2}}.
\]
Therefore, by \eqref{eq:lower_bound_2pH},
\[
    \frac{2^{H(D_{\ZZ^n,\sigma})}}{G(D_{\ZZ^n,\sigma})}
        \geqslant(1-e^{-1/2\sigma^2})\left(\frac{\sqrt{e}}{2a(\sigma)}\right)^n,
    \qquad
    a(\sigma)=e^{8\pi^2\sigma^2
            e^{-2\pi^2\sigma^2}\tanh(\pi^2\sigma^2)}.
\]
All asymptotics in what follows are as $\sigma\to\infty$. We introduce $\alpha=2\pi^2\sigma^2$
to simplify calculations.
Recall that $\tanh(\pi^2\sigma^2)=1-2e^{-\alpha}+o(e^{-\alpha})$. Hence,
\[
    8\pi^2\sigma^2
            e^{-2\pi^2\sigma^2}\tanh(\pi^2\sigma^2)
    =4\alpha e^{-\alpha}(1-2e^{-\alpha}+o(e^{-\alpha})).
\]
Note that this quantity goes to $0$ as $\alpha\to\infty$.
It follows that
\begin{align*}
    a(\sigma)
        &=e^{4\alpha e^{-\alpha}(1-2e^{-\alpha}+o(e^{-\alpha}))}\\
        &=1+4\alpha e^{-\alpha}(1-2e^{-\alpha}+o(e^{-\alpha}))+o(\alpha e^{-\alpha})\\
        &=1+4\alpha e^{-\alpha}+o(\alpha e^{-\alpha}).
\end{align*}
A similar analysis holds for $G^{qc}(D_{\ZZ^n,\sigma})$.
Recall that for $\sigma\geqslant \sqrt[4]{\frac{2}{27\pi^2}}$,
\[
    G^{qc}(D_{\ZZ^n,\sigma})
        \leqslant \frac{7}{6}\cdot\left(\frac{3}{2}\right)^{3n/4}
            \frac{\sqrt{\rho_{\sigma}(\ZZ^n)}}{(1-e^{-1/3\sigma^2})^{3/2}}.
\]
Therefore, for $\sigma\geqslant \sqrt[4]{\frac{2}{27\pi^2}}$, by \eqref{eq:lower_bound_2pH},
\begin{align*}
    \frac{2^{H(D_{\ZZ^n,\sigma})/2}}{G^{qc}(D_{\ZZ^n,\sigma})}
        \geqslant\frac{6}{7}(1-e^{-1/3\sigma^2})^{3/2}
        {\left(\frac{8e}{27a{(\sigma)}^2}\right)}^{n/4}
\end{align*}
where $a(\sigma)$ was defined above.

\subsection{Bounds on $G(D_{\ZZ_q^n,\sigma})$ and $G^{qc}(D_{\ZZ_q^n,\sigma})$}
We now consider the case of the modular discrete Gaussian. We do not know how to order the elements of $\ZZ_q^n$ by decreasing probability so
we will instead consider one possible order and bound the complexity of this order.
This will prove an upper bound on $G(D_{\ZZ_q^n,\sigma})$. Our proof
strategy is to relate the guessing complexity of $D_{\ZZ_q^n,\sigma}$ to
that of $D_{\ZZ^n,\sigma}$ and use the results proven in the previous
subsections.

Let $\tau:\ZZ_q^n\to\NN$ be an ordering of $\ZZ_q^n$ such that for all $\vec{x},\vec{y}\in\ZZ_q^n$,
if $\|\tilde{\vec{x}}\|<\|\tilde{\vec{y}}\|$ then $\tau(\vec{x})<\tau(\vec{y})$. In other words, we order points
of $\ZZ_q^n$ according to the norm of their ``lift'' in
$\set{\tfrac{q-1}{2},\ldots,\tfrac{q-1}{2}}$. Intuitively, when $\sigma$ is much smaller
than $q$, this will be the optimal order but we were not able to show this result.
We, however, do not require an optimal order to obtain an upper bound.
We now have that
\begin{align*}
    G(D_{\ZZ_q^n,\sigma})
        &\leqslant\sum_{x\in\ZZ_q^n}D_{\ZZ_q^n,\sigma}(x)\cdot \tau(x)\\
        &=\frac{1}{\rho_\sigma(\ZZ^n)}\cdot \sum_{\vec{x}\in\ZZ_q^n}\tau(\vec{x})\cdot \sum_{\vec{y}\in \ZZ^n}\rho_\sigma(\vec{x}+q\cdot\vec{y})\\
        &=\frac{1}{\rho_\sigma(\ZZ^n)}\sum_{\vec{x}\in\ZZ_q^n}
            \sum_{\vec{y}\in \ZZ^n}\rho_\sigma(\vec{x}+q\cdot\vec{y})\tau(\widetilde{\vec{x}+q\cdot\vec{y}})
            &&\text{since }\widetilde{\vec{x}+q\cdot\vec{y}}=\tilde{\vec{x}}\\
        &=\frac{1}{\rho_\sigma(\ZZ^n)}\cdot \sum_{\vec{x}\in\ZZ^n}\rho_\sigma(x)\cdot \tau(\widetilde{\vec{x}}).
\end{align*}
Now fix $\vec{x}\in\ZZ^n$. We now observe that by definition of the order $\tau$,
$\tau(\tilde{\vec{x}})<\tau(\tilde{\vec{y}})$ for any $\vec{y}\in\ZZ_q^n$ such that $\|\tilde{\vec{y}}\|>\|\tilde{\vec{x}}\|$.
In particular, choose $\vec{y}$ such that $\|\tilde{\vec{y}}\|=\|\tilde{\vec{x}}\|$ and $\tau(\vec{y})$
is the largest possible among all such $\vec{y}$. Then
\begin{align*}
    \tau(\tilde{\vec{y}})
        &=|\set{\vec{z}\in\ZZ_q^n:\|\tilde{\vec{z}}\|\leqslant \|\tilde{\vec{x}}\|}|\\
        &\leqslant |\set{\vec{z}\in\ZZ^n:\|\vec{z}\|\leqslant \|\tilde{\vec{x}}\|}|\\
        &=N(\|\tilde{\vec{x}}\|^2)
\end{align*}
where we recall that we defined $N(\ell)=\set{\vec{x}\in\ZZ^n:\norm{\vec{x}}\leqslant\sqrt{\ell}}$
and $N(-1)=0$ at the beginning of the proof. Therefore,
\begin{align*}
    G(D_{\ZZ_q^n,\sigma})
        &\leqslant \frac{1}{\rho_\sigma(\ZZ^n)}\sum_{\vec{x}\in\ZZ^n}\rho_\sigma(\vec{x})N(\|\tilde{\vec{x}}\|^2)\\
        &=\frac{1}{\rho_\sigma(\ZZ^n)}
            \sum_{\ell=0}^\infty\sum_{\vec{x}\in\ZZ^n:\|\vec{x}\|^2=\ell}
                \rho_\sigma(\vec{x})N(\|\tilde{\vec{x}}\|^2)\\
        &=\frac{1}{\rho_\sigma(\ZZ^n)}
            \sum_{\ell=0}^\infty\rho_\sigma(\sqrt{\ell})\cdot N(\ell)
                \sum_{\vec{x}\in\ZZ^n:\|\vec{x}\|^2=\ell}1\\
        &=\frac{1}{\rho_\sigma(\ZZ^n)}
            \sum_{\ell=0}^\infty\rho_\sigma(\sqrt{\ell})\cdot N(\ell) \cdot
                (N(\ell)-N(\ell-1))\\
        &\leqslant 2G(D_{\ZZ^n,\sigma})
            &&\text{by \eqref{eq:formula_G_X}.}
\end{align*}
We now consider the case of the quantum guessing complexity.
Virtually the same argument yields that
\begin{align*}
    G^{qc}(D_{\ZZ_q^n,\sigma})
        &\leqslant \frac{1}{\rho_\sigma(\ZZ^n)}\sum_{\vec{x}\in\ZZ^n}\rho_\sigma(\vec{x})
            \sqrt{N(\|\tilde{\vec{x}}\|^2)}\\
        &=\frac{1}{\rho_\sigma(\ZZ^n)}
            \sum_{\ell=0}^\infty\rho_\sigma(\sqrt{\ell})\cdot \sqrt{N(\ell)}
                \sum_{\vec{x}\in\ZZ^n:\|\vec{x}\|^2=\ell}1\\
        &=\frac{1}{\rho_\sigma(\ZZ^n)}
            \sum_{\ell=1}^\infty\rho_\sigma(\sqrt{\ell})\cdot \sqrt{N(\ell)}\cdot
                (N(\ell)-N(\ell-1)).
\end{align*}
Now recall by \eqref{eq:formula_Gqc_X} that
\[
    G^{qc}(D_{\ZZ^n,\sigma})
        =\sum_{\ell=0}^\infty
            \frac{\rho_\sigma(\sqrt{\ell})}{\rho_\sigma(\ZZ^n)}
                \sum_{i=N(\ell-1)+1}^{N(\ell)}\sqrt{i}.
\]
Furthermore, since $\sqrt{\cdot}$ is an increasing function, it is not hard to see that
for all $a,b\in\NN$,
\[
    \sum_{i=a+1}^b\sqrt{i}\geqslant\int_{a}^b\sqrt{x}\,\mathrm{d}x=\frac{2}{3}(b^{3/2}-a^{3/2}).
\]
Therefore, for any $\ell\in\NN$,
\[
    \sum_{i=N(\ell-1)+1}^{N(\ell)}\sqrt{i}
        \geqslant \frac{2}{3}\left(N(\ell)^{3/2}-N(\ell-1)^{3/2}\right).
\]
But check that $N(\ell-1)\leqslant N(\ell)$ so that
\[
    \sqrt{N(\ell)}(N(\ell)-N(\ell-1))
        \leqslant {N(\ell)}^{3/2}-{N(\ell-1)}^{3/2}
        \leqslant \frac{3}{2}\sum_{i=N(\ell-1)+1}^{N(\ell)}\sqrt{i}.
\]
It then easily follows that
\[
    G^{qc}(D_{\ZZ_q^n,\sigma})
        \leqslant \frac{3}{2}\frac{1}{\rho_\sigma(\ZZ^n)}
            \sum_{\ell=1}^\infty\rho_\sigma(\sqrt{\ell})
            \sum_{i=N(\ell-1)+1}^{N(\ell)}\sqrt{i}
        = \frac{3}{2}G^{qc}(D_{\ZZ^n,\sigma}).
\]
This finishes the proof. \qed{}

\section{Proofs for \cref{sec:quantum_dual_attack}}\label{sec:proofs-for-algorithm}

\subsection{Proof of \cref{th:analysis_quantum_alg}}

Below, we will establish the following claims:
\begin{enumerate}
\item[(1)]
  For all $\tilde{\vec{s}}_\enum$,
  the oracle $\hat{\mathcal{O}}$ inside $\mathcal{O}(\tilde{s}_\enum)$
  is such that, for all $\tilde{\vec{s}}_\fft$, with probability at least
  $9/10$, if
  $F_L(\tilde{\vec{s}}_\enum,\tilde{\vec{s}}_\fft)>(1+2\eta)\cdot C$ then
  $\hat{\mathcal{O}}(\tilde{\vec{s}}_\fft)=1$ and
  if $F_L(\tilde{\vec{s}}_\enum,\tilde{s}_\fft)\leqslant C$
  then
  $\hat{\mathcal{O}}(\tilde{\vec{s}}_\fft)=0$.
\item[(2)]
  For all $\tilde{\vec{s}}_\enum$, with probability at least $9/10$,
  if there exists $\tilde{\vec{s}}_\fft$ such that $F_L(\tilde{\vec{s}}_\enum,\tilde{\vec{s}}_\fft)>(1+2\eta)\cdot C$
  then $\mathcal{O}(\tilde{\vec{s}}_\enum)=1$.
\item[(3)]
  For all $\tilde{\vec{s}}_\enum$, with probability at least
  $9/10$,
  if $F_L(\tilde{\vec{s}}_\enum,\tilde{\vec{s}}_\fft)\leqslant C$ for all $\tilde{\vec{s}}_\fft$
  then $\mathcal{O}(\tilde{\vec{s}}_\enum)=0$.
\item[(4)] With probability at least $9/10$, if
  the algorithm returns $\tilde{\vec{s}}_\enum\neq\bot$
  then there exists $\tilde{\vec{s}}_\fft$ such that
  $F_L(\tilde{\vec{s}}_\enum,\tilde{\vec{s}}_\fft)>C$.
\end{enumerate}
We start by establishing the result as following from the claims and then establish these claims below. Let $x$ be the output of
the algorithm. By claim (4), if $x\neq\bot$ then, with probability at least $9/10$,
there exist $\tilde{\vec{s}}_\fft$ such that $F_L(\tilde{\vec{s}}_\enum,\tilde{\vec{s}}_\fft)>C$.
Therefore, $x\in S_{C}^{L}$. Hence, this proves that $x\in S_{C}^{L}\cup\set{\bot}$
with probability at least $9/10$.
Now assume that $S_{(1+2\eta)C}^{L}\neq\varnothing$ and let
$\tilde{\vec{s}}_\enum\in S_{(1+2\eta)C}^{L}$.
Then by claim (2), with probability at least $9/10$, $\mathcal{O}(\tilde{\vec{s}}_\enum)=1$
so the algorithm will not return $\bot$, i.e.~$x\neq\bot$.

\noindent \textbf{Proof of claim (1)}. Fix $\tilde{\vec{s}}_\enum$ and check that $\hat{\mathcal{O}}(\tilde{\vec{s}}_\fft)$ returns $1$ if and only if
\[
  \mathcal{A}^{\mathcal{O}_W'}((\tilde{\vec{s}}_\enum,\tilde{\vec{s}}_\fft,1))>(1+\eta)\cdot\tfrac{C}{D}.
\]
Now \(\mathcal{O}_W\) is defined in such a way that
\[
  \mathcal{O}_W'(j)=\left(\tfrac{p}{q}\cdot\vec{y}_{j,\enum},
    \left\lfloor \tfrac{p}{q}\cdot \vec{y}_{j,\fft}\right\rceil,\theta-\tfrac{p}{q}\cdot \vec{x}_j^T\cdot \vec{b} \right)
\]
where \(\vec{y}_{j,\enum}\) and \(\vec{y}_{j,\fft}\) are defined as expected. Therefore, by \Cref{thm:Quantum_estimation_f_W}, with probability at least $1-\delta$,
\[
  \left|\mathcal{A}^{\mathcal{O}_W'}((\tilde{\vec{s}}_{\enum}, \tilde{\vec{s}}_\fft,1))-f_W((\tilde{\vec{s}}_\enum,\tilde{\vec{s}}_\fft,1))\right|
  \leqslant\varepsilon.
\]
But one checks that
{\small\begin{align*}
         \left\langle \mathcal{O}_W(j),(\tilde{\vec{s}}_\enum, \tilde{\vec{s}}_\fft,1)\right\rangle
         &=\left\langle \left(\frac{p}{q}\cdot\vec{y}_{j,\enum},
           \round{\frac{p}{q}\cdot \vec{y}_{j,\fft}},\theta-\tfrac{p}{q}\cdot\vec{x}_j^T\cdot\vec{b}\right),
           (\tilde{\vec{s}}_\enum, \tilde{\vec{s}}_\fft,1)\right\rangle\\
         &= \frac{p}{q}\cdot\vec{y}_{j,\enum}^T\cdot \tilde{\vec{s}}_\enum
           +\round{\frac{p}{q}\cdot\vec{y}_{j,\fft}}^T\cdot \tilde{\vec{s}}_\fft
           + \theta-\frac{p}{q}\cdot\vec{x}_j^T\cdot \vec{b}.
\end{align*}}
Therefore, \(f_W((\tilde{\vec{s}}_\enum,\tilde{\vec{s}}_\fft,1))\)
{\small\begin{align*}
              &=\frac{1}{D}\,\sum_{j}\cos\left(\frac{2\pi}{p}\left\langle \mathcal{O}_W(j),(\tilde{\vec{s}}_\enum, \tilde{\vec{s}}_\fft,1)\right\rangle\right)\\
              &=\frac{1}{D}\sum_{j}\cos\left(\frac{2\pi}{p}\left(
                \frac{p}{q}\cdot\vec{y}_{j,\enum}^T\cdot \tilde{\vec{s}}_\enum
                +\round{\frac{p}{q}\cdot\vec{y}_{j,\fft}}^T\cdot\tilde{\vec{s}}_\fft
                + \theta-\frac{p}{q}\cdot\vec{x}_j^T\cdot \vec{b}\right)+\frac{2\pi}{p}\cdot\theta\right)\\
            &=\frac{1}{D}\,
                \Re{\left(
                    \sum_{j}\exp\left(
                        \frac{2i\pi}{p}\left(
                        \frac{p}{q}\cdot \vec{y}_{j,\enum}^T\cdot \tilde{\vec{s}}_\enum
                        + \round{\frac{p}{q}\cdot \vec{y}_{j,\fft}}^T\cdot\tilde{\vec{s}}_\fft
                        -\frac{p}{q}\cdot \vec{x}_j^T\cdot \vec{b}\right)+\frac{2i\pi}{p}\cdot\theta
                    \right)
                \right)}\\
            &=\frac{1}{D}\,
                \Re{\left(e^{\frac{2i\pi}{p}\theta}
                    \sum_{j}\exp\left(
                        \frac{2i\pi}{p}\left(
                        \frac{p}{q}\cdot\vec{y}_{j,\enum}^T\cdot\tilde{\vec{s}}_\enum
                        +\round{\frac{p}{q}\cdot\vec{y}_{j,\fft}}^T\cdot\tilde{\vec{s}}_\fft
                        -\frac{p}{q}\cdot \vec{x}_j^T\cdot\vec{b}\right)
                    \right)
                \right)}\\
            &=\frac{1}{D}\, F_L(\tilde{\vec{s}}_\enum,\tilde{\vec{s}}_\fft)
\end{align*}}
since $\theta$ was computed so that $\psi(\tilde{\vec{s}}_\fft)=e^{-\tfrac{2i\pi}{p}\theta}$.
It follows that, with probability at least $1-\delta$,
\[
  \left|\mathcal{A}^{\mathcal{O}_W'}((\tilde{\vec{s}}_{\enum}, \tilde{\vec{s}}_\fft,1))-\frac{1}{D}F_L(\tilde{\vec{s}}_\enum,\tilde{\vec{s}}_\fft)\right|
  \leqslant \varepsilon=\frac{C}{D}\cdot\eta.
\]
Assume that this inequality holds.
\begin{itemize}
\item If $F_L(\tilde{\vec{s}}_\enum,\tilde{\vec{s}}_\fft)>(1+2\eta)\cdot C$ then
  $\mathcal{A}^{\mathcal{O}_W'}((\tilde{\vec{s}}_\enum,\tilde{\vec{s}}_\fft,1))>(1+2\eta)\cdot\frac{C}{D}-\frac{C}{D}\cdot\eta
  =(1+\eta)\cdot \frac{C}{D}$
  so $\hat{\mathcal{O}}(\tilde{\vec{s}}_\fft)=1$.
\item If $F_L(\tilde{\vec{s}}_\enum,\tilde{\vec{s}}_\fft)\leqslant C$ then
  $\mathcal{A}^{\mathcal{O}_W'}((\tilde{\vec{s}}_\enum,\tilde{\vec{s}}_\fft,1))\leqslant
  \frac{C}{D}+\frac{C}{D}\eta=(1+\eta)\frac{C}{D}$
  so $\hat{\mathcal{O}}(\tilde{\vec{s}}_\fft)=0$.
\end{itemize}

\noindent \textbf{Proof of claim (2)}. Fix $\tilde{\vec{s}}_\enum$. If there exists $\tilde{\vec{s}}_\fft$ such that $F_L(\tilde{\vec{s}}_\enum,\tilde{\vec{s}}_\fft)>(1+2\eta)\cdot C$ then by claim (1), with probability at least
$1-\delta$, $\hat{\mathcal{O}}(\tilde{\vec{s}}_\fft)=1$. It follows by \Cref{thm:hoyer03}
that the search will, with probability at least $9/10$, return $i\neq\bot$
and therefore $\mathcal{O}(\tilde{\vec{s}}_\enum)$ will return 1.

\noindent \textbf{Proof of claim (3)}. For $\mathcal{O}(\tilde{\vec{s}}_\enum)$ to return $0$,
it is sufficient to have $\hat{\mathcal{O}}(\tilde{\vec{s}}_\fft)=0$ for all $\tilde{\vec{s}}_\fft$.
By claim (1), $\hat{\mathcal{O}}(\tilde{\vec{s}}_\fft)=0$ with probability at least $1-\delta$
when $F_L(\tilde{\vec{s}}_\enum,\tilde{\vec{s}}_\fft)\leqslant C$. Hence, by \Cref{thm:hoyer03},
the search algorithm will return $\bot$ with probability $9/10$ and
$\mathcal{O}(\tilde{\vec{s}}_\enum)=0$. Note here that there is no need for a union bound
because of \Cref{thm:hoyer03}.

\noindent \textbf{Proof of claim (4)}. For the algorithm to return $\tilde{\vec{s}}_\enum$, with probability
$9/10$, we must have $\mathcal{O}(\tilde{\vec{s}}_\enum)=1$. By claim (3), with probability at
least
$9/10$,
this can only happen if
$F_L(\tilde{\vec{s}}_\enum,\tilde{\vec{s}}_\fft)>C$ for some $\tilde{\vec{s}}_\fft$.
Therefore the probability that the algorithm returns $\tilde{\vec{s}}_\enum$ such that
$F_L(\tilde{\vec{s}}_\enum,\tilde{\vec{s}}_\fft)\leqslant C$ for all $\tilde{\vec{s}}_\fft$
is bounded by $1/10$, by \Cref{lem:quantum_find_first}.
This finishes the proof. \qed{}

\subsection{Proof of \cref{lem:correctness_quantum}}

Recall that for any list $L$ and any $x>0$, we defined

\[
  S_x^L=\set{\tilde{\vec{s}}_\enum:\exists\ \tilde{\vec{s}}_\fft,
    F_L(\tilde{\vec{s}}_\enum,\tilde{\vec{s}}_\fft)>x}.
\]
In the proof of~\cite[Theorem~5.2]{Matzov22}, it is shown that for any threshold\footnote{The proof assumes a particular value of $C$ but the first three lines of the derivation in~\cite[Theorem~5.2]{Matzov22} hold for any value of $C$, which we call $X$ here.} $X$,
\[
  \Prob{L}{F_L(\vec{s}_\enum,\vec{s}_\fft)>X}
  \geq \Phi\left(\phi_{\mathrm{fp}}(\mu)
    +\phi_{\mathrm{fn}}(\mu) -\frac{X}{\sqrt{D_{\text{arg}}\cdot D}}\right).
\]
and that for any $\tilde{\vec{s}}_\enum\neq \vec{s}_\enum$, any $\tilde{\vec{s}}_\fft$
and any threshold $Y$,
\[
  \Prob{L}{F_L(\tilde{\vec{s}}_\enum,\tilde{\vec{s}}_\fft)>Y}
  \leqslant 1-\Phi\left(\frac{Y}{\sqrt{D_{\text{arg}}\cdot D}}\right).
\]
We are going to apply those inequalities to $X=(1+2\eta)\cdot C$ and $Y=C$.
The second inequality, by the choice of $C$, gives that
\[
  \Prob{L}{F_L(\tilde{\vec{s}}_\enum,\tilde{\vec{s}}_\fft)>C}
  \leqslant 1-\Phi\left(\frac{C}{\sqrt{D_{\text{arg}}\cdot D}}\right)
  =\frac{\mu}{2N_\enum(\vec{s}_\enum)\cdot p^{k_\fft}}.
\]
The number of guesses of $\tilde{\vec{s}}_\enum$ before reaching $\vec{s}_\enum$ is $N_\enum(\vec{s}_\enum)$
in the classical case. However note that \Cref{lem:quantum_find_first} may call
the oracle on more entries than $N_\enum(\vec{s}_\enum)$. Specifically,
\Cref{lem:quantum_find_first} guarantees that the oracle will only call the oracle
on the first $2N_\enum(\vec{s}_\enum)$ entries (with constant probability).
Therefore, by a union bound,
\begin{equation}\label{eq:bad_case}
  \Prob{L}{F_L(\tilde{\vec{s}}_\enum,\tilde{\vec{s}}_\fft)\leqslant C\text{ for the first }2N_\enum(\vec{s}_\enum)\text{ values of }\tilde{\vec{s}}_\enum}
  \geq1-\mu.
\end{equation}
On the other hand,
\begin{align*}
  \Prob{L}{F_L(\vec{s}_\enum,\vec{s}_\fft)>(1+2\eta)C}
  &\geq \Phi\left(\phi_{\text{fp}}(\mu)+\phi_{\text{fn}}(\mu)
    -(1+2\eta)\frac{C}{\sqrt{D_{\text{arg}}\cdot D}}\right)\\
  &=\Phi\left(\phi_{\text{fp}}(\mu)+\phi_{\text{fn}}(\mu)
    -(1+2\eta)\phi_{\text{fp}}(\mu)\right)\\
  &=\Phi\left(\phi_{\text{fn}}(\mu)
    -2\eta\phi_{\text{fp}}(\mu)\right).
\end{align*}
It is easy to check by taking the derivative that $\Phi$, the cdf of the normal distribution, satisfies the following
inequality for all $y\geq 0$:
\[
  \Phi(y)\geq 1-\frac{e^{-y^2/2}}{\sqrt{\pi}}.
\]
Furthermore, $\Phi$ is an increasing function so $\Phi^{-1}$ is also increasing.
Hence,
\[
  \Phi^{-1}(1-x)\leqslant \sqrt{-2\ln(\pi x)}
\]
for all $x\leqslant\frac{1}{\sqrt{\pi}}$. Now recall that
\[
  \phi_{\text{fp}}(\mu)=\Phi^{-1}\left(1-\frac{\mu}{2N_\enum(\vec{s}_\enum)p^{k_\fft}}\right),
  \qquad
  \phi_{\text{fn}}(\mu)=\Phi^{-1}\left(1-\frac{\mu}{2}\right).
\]
Therefore,
\[
  \phi_{\text{fp}}(\mu)
  \leqslant \sqrt{-2\ln\frac{\pi\mu}{2N_\enum(\vec{s}_\enum)p^{k_\fft}}}.
\]
From this we get that $\phi_{\text{fp}}(\mu)$ is a polynomial factor in all the relevant parameters. Now observe that by the integral definition of $\Phi()$,
\begin{align*}
  \Phi\left(\phi_{\text{fn}}(\mu)-2\eta\phi_{\text{fp}}(\mu)\right)
  &=\Phi\left(\phi_{\text{fn}}(\mu)\right)
    -\frac{1}{\sqrt{2\pi}}
    \int_{\phi_{\text{fn}}(\mu)-2\eta\phi_{\text{fp}}(\mu)}^{\phi_{\text{fn}}(\mu)}
    e^{-t^2/2}\mathrm{d}t\\
  &\geq
    1-\frac{\mu}{2}-\frac{2\eta}{\sqrt{2\pi}}\phi_{\text{fp}}(\mu)\\
  &\geq
    1-\frac{3\mu}{4}
\end{align*}
when
\[
  \eta\leqslant\frac{\sqrt{2\pi}\mu}{8\phi_{\text{fp}}(\mu)}.
\]
Therefore, by \Cref{th:analysis_quantum_alg}, with probability at least
$1-\frac{3\mu}{4}$, we have that $S^L_{(1+2\eta)C}\neq\varnothing$ so the algorithm
returns an element from $S^L_C$.
Furthermore, by \Cref{eq:bad_case}, with probability at least $1-\mu$,
this element must be
$\vec{s}_\enum$ because all the other elements satisfy
$F_L(\tilde{\vec{s}}_\enum,\tilde{\vec{s}}_\fft)\leqslant C$.
Therefore, by a union bound, the probability that the algorithm returns $\vec{s}_\enum$
is at least $1-\tfrac{3\mu}{4}-\mu\geq1-2\mu\geq1-\nu$.
This concludes the proof. \qed{}

\subsection{Proof of \cref{lem:correctness_quantum}}

Recall that for any list $L$ and any $x>0$, we defined

\[
  S_x^L=\set{\tilde{\vec{s}}_\enum:\exists\ \tilde{\vec{s}}_\fft,
    F_L(\tilde{\vec{s}}_\enum,\tilde{\vec{s}}_\fft)>x}.
\]
In the proof of~\cite[Theorem~5.2]{Matzov22}, it is shown that for any threshold\footnote{The proof assumes a particular value of $C$ but the first three lines of the derivation in~\cite[Theorem~5.2]{Matzov22} hold for any value of $C$, which we call $X$ here.} $X$,
\[
  \Prob{L}{F_L(\vec{s}_\enum,\vec{s}_\fft)>X}
  \geq \Phi\left(\phi_{\mathrm{fp}}(\mu)
    +\phi_{\mathrm{fn}}(\mu) -\frac{X}{\sqrt{D_{\text{arg}}\cdot D}}\right).
\]
and that for any $\tilde{\vec{s}}_\enum\neq \vec{s}_\enum$, any $\tilde{\vec{s}}_\fft$
and any threshold $Y$,
\[
  \Prob{L}{F_L(\tilde{\vec{s}}_\enum,\tilde{\vec{s}}_\fft)>Y}
  \leqslant 1-\Phi\left(\frac{Y}{\sqrt{D_{\text{arg}}\cdot D}}\right).
\]
We are going to apply those inequalities to $X=(1+2\eta)\cdot C$ and $Y=C$.
The second inequality, by the choice of $C$, gives that
\[
  \Prob{L}{F_L(\tilde{\vec{s}}_\enum,\tilde{\vec{s}}_\fft)>C}
  \leqslant 1-\Phi\left(\frac{C}{\sqrt{D_{\text{arg}}\cdot D}}\right)
  =\frac{\mu}{2N_\enum(\vec{s}_\enum)\cdot p^{k_\fft}}.
\]
The number of guesses of $\tilde{\vec{s}}_\enum$ before reaching $\vec{s}_\enum$ is $N_\enum(\vec{s}_\enum)$
in the classical case. However note that \Cref{lem:quantum_find_first} may call
the oracle on more entries than $N_\enum(\vec{s}_\enum)$. Specifically,
\Cref{lem:quantum_find_first} guarantees that the oracle will only call the oracle
on the first $2N_\enum(\vec{s}_\enum)$ entries (with constant probability).
Therefore, by a union bound,
\begin{equation}\label{eq:bad_case}
  \Prob{L}{F_L(\tilde{\vec{s}}_\enum,\tilde{\vec{s}}_\fft)\leqslant C\text{ for the first }2N_\enum(\vec{s}_\enum)\text{ values of }\tilde{\vec{s}}_\enum}
  \geq1-\mu.
\end{equation}
On the other hand,
\begin{align*}
  \Prob{L}{F_L(\vec{s}_\enum,\vec{s}_\fft)>(1+2\eta)C}
  &\geq \Phi\left(\phi_{\text{fp}}(\mu)+\phi_{\text{fn}}(\mu)
    -(1+2\eta)\frac{C}{\sqrt{D_{\text{arg}}\cdot D}}\right)\\
  &=\Phi\left(\phi_{\text{fp}}(\mu)+\phi_{\text{fn}}(\mu)
    -(1+2\eta)\phi_{\text{fp}}(\mu)\right)\\
  &=\Phi\left(\phi_{\text{fn}}(\mu)
    -2\eta\phi_{\text{fp}}(\mu)\right).
\end{align*}
It is easy to check by taking the derivative that $\Phi$, the cdf of the normal distribution, satisfies the following
inequality for all $y\geq 0$:
\[
  \Phi(y)\geq 1-\frac{e^{-y^2/2}}{\sqrt{\pi}}.
\]
Furthermore, $\Phi$ is an increasing function so $\Phi^{-1}$ is also increasing.
Hence,
\[
  \Phi^{-1}(1-x)\leqslant \sqrt{-2\ln(\pi x)}
\]
for all $x\leqslant\frac{1}{\sqrt{\pi}}$. Now recall that
\[
  \phi_{\text{fp}}(\mu)=\Phi^{-1}\left(1-\frac{\mu}{2N_\enum(\vec{s}_\enum)p^{k_\fft}}\right),
  \qquad
  \phi_{\text{fn}}(\mu)=\Phi^{-1}\left(1-\frac{\mu}{2}\right).
\]
Therefore,
\[
  \phi_{\text{fp}}(\mu)
  \leqslant \sqrt{-2\ln\frac{\pi\mu}{2N_\enum(\vec{s}_\enum)p^{k_\fft}}}.
\]
From this we get that $\phi_{\text{fp}}(\mu)$ is a polynomial factor in all the relevant parameters. Now observe that by the integral definition of $\Phi()$,
\begin{align*}
  \Phi\left(\phi_{\text{fn}}(\mu)-2\eta\phi_{\text{fp}}(\mu)\right)
  &=\Phi\left(\phi_{\text{fn}}(\mu)\right)
    -\frac{1}{\sqrt{2\pi}}
    \int_{\phi_{\text{fn}}(\mu)-2\eta\phi_{\text{fp}}(\mu)}^{\phi_{\text{fn}}(\mu)}
    e^{-t^2/2}\mathrm{d}t\\
  &\geq
    1-\frac{\mu}{2}-\frac{2\eta}{\sqrt{2\pi}}\phi_{\text{fp}}(\mu)\\
  &\geq
    1-\frac{3\mu}{4}
\end{align*}
when
\[
  \eta\leqslant\frac{\sqrt{2\pi}\mu}{8\phi_{\text{fp}}(\mu)}.
\]
Therefore, by \Cref{th:analysis_quantum_alg}, with probability at least
$1-\frac{3\mu}{4}$, we have that $S^L_{(1+2\eta)C}\neq\varnothing$ so the algorithm
returns an element from $S^L_C$.
Furthermore, by \Cref{eq:bad_case}, with probability at least $1-\mu$,
this element must be
$\vec{s}_\enum$ because all the other elements satisfy
$F_L(\tilde{\vec{s}}_\enum,\tilde{\vec{s}}_\fft)\leqslant C$.
Therefore, by a union bound, the probability that the algorithm returns $\vec{s}_\enum$
is at least $1-\tfrac{3\mu}{4}-\mu\geq1-2\mu\geq1-\nu$.
This concludes the proof. \qed{}

\section{Quantum algorithm w/o QRACM for the FFT threshold problem}\label{sec:open_question_calc}

Here, for simplicity, we assume that we want
to decide if
$|{f_L}(\vec{x})|>\delta^+$, under the promise
that $|{f_L}(\vec{x})|\notin[\delta^-,\delta^+]$.
In the original problem, we were using the real part.
We focus on the \PromiseSparseFFTThreshold{} problem. Assume that we can create the superposition
\begin{equation}\label{eq:superposition_for_fft}
    \ket{\psi}=\frac{1}{\sqrt{Z}}\sum_jw_j\ket{\vec{u}_j}.
\end{equation}
We then perform a QFT on $\ket{\psi}$ to obtain
\[
    \ket{\widehat{\psi}}=\frac{1}{\sqrt{Z\cdot |G|}}\sum_{\vec{x}\in G}{f_L}(\vec{x})\ket{\vec{x}}.
\]
If we measure $\ket{\widehat{\psi}}$, we obtain $\vec{x}$ with probability
\[
    p(\vec{x})=\frac{|\widehat{f_L}(\vec{x})|^2}{|G|\cdot Z}.
\]
Consider the random variable $X$ defined by $\Prob{}{X=\vec{x}}=p(\vec{x})$.
Let us reformulate the problem:
\begin{itemize}
    \item We have access to a sampler for $X$.
    \item We do not know the values of $p(\vec{x})$.
    \item $p(\vec{x})\notin[\eta^-,\eta^+]$ for some given $\eta^-,\eta^+$.
    \item We want to decide if there exists $\vec{x}\in G$ such that $p(\vec{x})>\eta^+$.
\end{itemize}
We can relate the new parameters to the old one as follows:
\begin{align*}
    |{f_L}(\vec{x})|\notin[\delta^-,\delta^+]
    \;
    &\Rightarrow
    \;
    \frac{|{f_L}(\vec{x})|^2}{|G|\cdot Z}
        \notin\left[
            \frac{(\delta^-)^2}{Z\cdot|G|},
            \frac{(\delta^+)^2}{Z\cdot|G|}
            \right]
    \\
    &\Rightarrow
    \;
    \eta^{\pm}=\frac{(\delta^\pm)^2}{Z\cdot|G|}.
\end{align*}
We now consider the following algorithm:
\begin{itemize}
    \item sample $\vec{x}_0,\ldots,\vec{x}_{N-1}$ from $X$,
    \item compute $\tilde{p}(\vec{x})=\tfrac{1}{N}\left|\set{j:x_j=\vec{x}}\right|$ for every $\vec{x}$,
    \item check if $\tilde{p}(\vec{x})>(\eta^-+\eta^+)/2$ for some $\vec{x}$.
\end{itemize}
It is clear that with careful optimisations, this algorithm runs in time $O(N)$, up to some small
polynomial factors. We now need to analyse the success probability of this algorithm.
By the Dvoretzky–Kiefer–Wolfowitz inequality \cite{DKW_ineq},
there exists a constant $A$ such that for all $\varepsilon>0$,
\[
    \Prob{}{\sup_{\vec{x}\in G}|p(\vec{x})-\tilde{p}(\vec{x})|>\varepsilon}
        \leqslant A\cdot e^{-2N\varepsilon^2}.
\]
Here the probability is over the choice of the $\vec{x}_0,\ldots,\vec{x}_{N-1}$.
A result by Massart \cite{DKW_Massart} shows that $A=2$. For the algorithm to work, we need to
estimate $p(\vec{x})$ within $\pm(\eta^+-\eta^-)/2$ so we let $\varepsilon=(\eta^+-\eta^-)/2$.
If want our algorithm to succeed with probability at least $1-\nu$, we need to choose $N$ such
that
\[
    A\cdot e^{-2N\varepsilon^2}=\nu
    \qquad
    \Leftrightarrow
    \qquad
    N=\frac{\ln(A)-\ln(\nu)}{2\varepsilon^2}.
\]
\bigskip

 \section{Source code}\label{sec:source-code}

 Our code relies on the modified LWE Estimator from~\cite{JMC:AlbPlaSco15} available at \url{https://github.com/malb/lattice-estimator/}. We also \textattachfile{code/estimates.py}{attached our code} as an attachment to this PDF\@. Not all PDF viewers support this feature. If the reader's PDF reader does not then e.g.~\texttt{pdfdetach} can be used to extract the source code without having to copy and paste it by hand. To run our code, run \lstinline{git clone https://github.com/malb/lattice-estimator/} in the directory where \lstinline{estimates.py} is located.

\begin{lstlisting}[language=python]
# -*- coding: utf-8 -*-
"""
Run like this::

    sage: attach("estimates.py")
    sage: %time results = runall()
    sage: save(results, "../data/estimates.sobj")
    sage: print(results_table(results))

"""
from sage.all import sqrt, log, exp, e, pi, RR, ZZ

from estimator.estimator.cost import Cost
from estimator.estimator.lwe_parameters import LWEParameters
from estimator.estimator.reduction import delta as deltaf
from estimator.estimator.reduction import RC, ReductionCost
from estimator.estimator.conf import red_cost_model as red_cost_model_default
from estimator.estimator.util import local_minimum, early_abort_range
from estimator.estimator.io import Logging
from estimator.estimator.schemes import (
    Kyber512,
    Kyber768,
    Kyber1024,
    LightSaber,
    Saber,
    FireSaber,
)
from estimator.estimator.schemes import TFHE630, TFHE1024


class ChaLoy21(ReductionCost):

    __name__ = "ChaLoy21"
    short_vectors = ReductionCost._short_vectors_sieve

    def __call__(self, beta, d, B=None):
        """
        :param beta: Block size ≥ 2.
        :param d: Lattice dimension.
        :param B: Bit-size of entries.
        """

        return ZZ(2) ** RR(0.2570 * beta)


class MATZOV:
    """ """

    C_prog = 1.0 / (1 - 2.0 ** (-0.292))  # p.37
    C_mul = 32**2  # p.37
    C_add = 5 * 32  # guessing based on C_mul

    @classmethod
    def T_fftf(cls, k, p):
        """
        The time complexity of the FFT in dimension `k` with modulus `p`.

        :param k: Dimension
        :param p: Modulus ≥ 2

        """
        return cls.C_mul * k * p ** (k + 1)  # Theorem 7.6, p.38

    @classmethod
    def T_tablef(cls, D):
        """
        Time complexity of updating the table in each iteration.

        :param D: Number of nonzero entries

        """
        return 4 * cls.C_add * D  # Theorem 7.6, p.39

    @classmethod
    def Nf(cls, params, m, beta_bkz, beta_sieve, k_enum, k_fft, p):
        """
        Required number of samples to distinguish with advantage.

        :param params: LWE parameters
        :param m:
        :param beta_bkz: Block size used for BKZ reduction
        :param beta_sieve: Block size used for sampling
        :param k_enum: Guessing dimension
        :param k_fft: FFT dimension
        :param p: FFT modulus

        """
        mu = 0.5
        k_lat = params.n - k_fft - k_enum  # p.15

        # p.39
        lsigma_s = (
            params.Xe.stddev ** (m / (m + k_lat))
            * (params.Xs.stddev * params.q) ** (k_lat / (m + k_lat))
            * sqrt(4 / 3.0)
            * sqrt(beta_sieve / 2 / pi / e)
            * deltaf(beta_bkz) ** (m + k_lat - beta_sieve)
        )

        # p.29, we're ignoring O()
        N = (
            exp(4 * (lsigma_s * pi / params.q) ** 2)
            * exp(k_fft / 3.0 * (params.Xs.stddev * pi / p) ** 2)
            * (k_enum * cls.Hf(params.Xs) + k_fft * log(p) + log(1 / mu))
        )

        return RR(N)

    @staticmethod
    def Hf(Xs):
        return RR(1 / 2 + log(sqrt(2 * pi) * Xs.stddev)) / log(2.0)

    @classmethod
    def cost(
        cls,
        beta,
        params,
        m=None,
        p=2,
        k_enum=0,
        k_fft=0,
        beta_sieve=None,
        red_cost_model=red_cost_model_default,
    ):
        """
        Theorem 7.6

        """

        if m is None:
            m = params.n

        k_lat = params.n - k_fft - k_enum  # p.15

        # We assume here that β_sieve ≈ β
        N = cls.Nf(
            params,
            m,
            beta,
            beta_sieve if beta_sieve else beta,
            k_enum,
            k_fft,
            p,
        )
        rho, T_sample, _, beta_sieve = red_cost_model.short_vectors(
            beta, N=N, d=k_lat + m, sieve_dim=beta_sieve
        )

        H = cls.Hf(params.Xs)
        T_guess = (
            ((2 / sqrt(e)) ** k_enum)
            * (2 ** (k_enum * H))
            * (cls.T_fftf(k_fft, p) + cls.T_tablef(N))
        )
        cost = Cost(rop=T_sample + T_guess, problem=params)
        cost["red"] = T_sample
        cost["guess"] = T_guess
        cost["beta"] = beta
        cost["p"] = p
        cost["zeta"] = k_enum
        cost["t"] = k_fft
        cost["beta_"] = beta_sieve
        cost["N"] = N
        cost["m"] = m

        cost.register_impermanent(
            {"β'": False, "ζ": False, "t": False}, rop=True, p=False, N=False
        )
        return cost

    def __call__(
        self,
        params: LWEParameters,
        red_cost_model=red_cost_model_default,
        log_level=1,
    ):
        """
        Optimizes cost of dual attack as presented in [Matzov22]_.

        :param params: LWE parameters
        :param red_cost_model: How to cost lattice reduction

        The returned cost dictionary has the following entries:

        - ``rop``: Total number of word operations (≈ CPU cycles).
        - ``red``: Number of word operations in lattice reduction and
                   short vector sampling.
        - ``guess``: Number of word operations in guessing and FFT.
        - ``β``: BKZ block size.
        - ``ζ``: Number of guessed coordinates.
        - ``t``: Number of coordinates in FFT part mod `p`.
        - ``d``: Lattice dimension.

        """
        params = params.normalize()

        for p in early_abort_range(2, params.q):
            for k_enum in early_abort_range(0, params.n, 5):
                for k_fft in early_abort_range(0, params.n - k_enum[0], 5):
                    with local_minimum(
                        40, params.n, log_level=log_level + 4
                    ) as it:
                        for beta in it:
                            cost = self.cost(
                                beta,
                                params,
                                p=p[0],
                                k_enum=k_enum[0],
                                k_fft=k_fft[0],
                                red_cost_model=red_cost_model,
                            )
                            it.update(cost)
                        Logging.log(
                            "dual",
                            log_level + 3,
                            f"t: {k_fft[0]}, {repr(it.y)}",
                        )
                        k_fft[1].update(it.y)
                Logging.log(
                    "dual", log_level + 2, f"ζ: {k_enum[0]}, {repr(k_fft[1].y)}"
                )
                k_enum[1].update(k_fft[1].y)
            Logging.log("dual", log_level + 1, f"p:{p[0]}, {repr(k_enum[1].y)}")
            p[1].update(k_enum[1].y)
        Logging.log("dual", log_level, f"{repr(p[1].y)}")
        return p[1].y


class QMATZOV(MATZOV):
    @classmethod
    def cost(
        cls,
        beta,
        params,
        m=None,
        p=2,
        k_enum=0,
        k_fft=0,
        beta_sieve=None,
        red_cost_model=red_cost_model_default,
    ):
        """
        Theorem 7.6

        """

        if m is None:
            m = params.n

        k_lat = params.n - k_fft - k_enum  # p.15

        # We assume here that β_sieve ≈ β
        N = cls.Nf(
            params,
            m,
            beta,
            beta_sieve if beta_sieve else beta,
            k_enum,
            k_fft,
            p,
        )
        rho, T_sample, _, beta_sieve = red_cost_model.short_vectors(
            beta, N=N, d=k_lat + m, sieve_dim=beta_sieve
        )

        H = cls.Hf(params.Xs)
        T_guess = ((27 / 8 / e) ** (k_enum / 4)) * sqrt(
            2 ** (k_enum * H) * p ** k_fft * cls.T_tablef(N)
        )
        cost = Cost(rop=T_sample + T_guess, problem=params)
        cost["red"] = T_sample
        cost["guess"] = T_guess
        cost["beta"] = beta
        cost["p"] = p
        cost["zeta"] = k_enum
        cost["t"] = k_fft
        cost["beta_"] = beta_sieve
        cost["N"] = N
        cost["m"] = m

        cost.register_impermanent(
            {"β'": False, "ζ": False, "t": False}, rop=True, p=False, N=False
        )
        return cost


def runall(
    schemes=(
        Kyber512,
        Kyber768,
        Kyber1024,
        LightSaber,
        Saber,
        FireSaber,
        TFHE630,
        TFHE1024,
    ),
    # schemes=(Kyber512, Kyber768, Kyber1024, LightSaber, Saber, FireSaber),
    nns=(
        "list_decoding-naive_classical",
        "list_decoding-classical",
        "list_decoding-naive_quantum",
        "list_decoding-ge19",
    ),
):

    results = {}

    for scheme in schemes:
        results[scheme] = {}
        print(f"{repr(scheme)}")
        for nn in nns:
            cost = MATZOV()(scheme, red_cost_model=RC.MATZOV.__class__(nn=nn))
            results[scheme][(nn, "classical")] = cost
            print(f" nn: {nn},  cost: {repr(cost)}")
            cost = QMATZOV()(scheme, red_cost_model=RC.MATZOV.__class__(nn=nn))
            print(f" nn: {nn}, qcost: {repr(cost)}")
            results[scheme][(nn, "quantum")] = cost

        cost = MATZOV()(scheme, red_cost_model=RC.ADPS16)
        print(f" C0, cost: {repr(cost)}")
        results[scheme][("C0", "classical")] = cost

        cost = MATZOV()(scheme, red_cost_model=ChaLoy21())
        print(f" Q0, cost: {repr(cost)}")
        results[scheme][("Q0", "classical")] = cost

        cost = QMATZOV()(scheme, red_cost_model=ChaLoy21())
        print(f"Q0, qcost: {repr(cost)}")
        results[scheme][("Q0", "quantum")] = cost

    return results


def results_table(results, fmt=None):
    import tabulate

    rows = []

    def pp(cost):
        return round(log(cost["rop"], 2), 1)

    for scheme, costs in results.items():
        row = [
            scheme.tag,
            pp(costs[("list_decoding-classical", "classical")]),
            pp(costs[("list_decoding-naive_classical", "classical")]),
            pp(costs[("C0", "classical")]),
            pp(costs[("list_decoding-ge19", "classical")]),
            pp(costs[("list_decoding-naive_quantum", "classical")]),
            pp(costs[("Q0", "classical")]),
            pp(costs[("list_decoding-naive_quantum", "quantum")]),
            pp(costs[("Q0", "quantum")]),
        ]
        rows.append(row)
    if fmt is None:
        return rows
    else:
        return tabulate.tabulate(
            results_table(results),
            headers=[
                "Scheme",
                "CC",
                "CN",
                "C0",
                "GE19",
                "QN",
                "Q0",
                "This work (QN)",
                "This work (Q0)",
            ],
            tablefmt="latex_booktabs",
            floatfmt=".1f",
        )
\end{lstlisting}

\end{document}